\documentclass[journal]{IEEEtran}

\usepackage{setspace}
\usepackage{lettrine}

\usepackage{graphics,
	psfrag,
	epsfig,
	amsthm,
	cite,
	amssymb,
	url,
	dsfont,
	subfigure,
	balance,
	enumerate,
	color,
	setspace
}

\usepackage{algorithm} 
\usepackage{algorithmicx}
\usepackage{algpseudocode}

\usepackage{amsmath}

\usepackage{epstopdf}

\newtheorem{definition}{Definition}

\newtheorem{theorem}{Theorem}

\newtheorem{proposition}{Proposition}

\usepackage{comment}
\usepackage{derivative}

\newcommand{\eref}[1]{(\ref{#1})}
\newcommand{\sref}[1]{Section~\ref{#1}}

\newcommand{\dref}[1]{Definition~\ref{#1}}
\newcommand{\pref}[1]{Proposition~\ref{#1}}
\newcommand{\cref}[1]{Constraint~\ref{#1}}

\newcommand{\tref}[1]{Table~\ref{#1}}
\newcommand{\Fref}[1]{Fig.~\ref{#1}}

\newcommand{\ignore}[1]{}

\IEEEoverridecommandlockouts
\usepackage{mathtools}
\usepackage{color}
\usepackage[utf8]{inputenc}

\begin{document}
\IEEEoverridecommandlockouts


\title{\vspace{-.9cm} Effectiveness of Reconfigurable Intelligent Surfaces to Enhance Connectivity in UAV Networks}

\author{Mohammed S. Al-Abiad, Member, IEEE, Mohammad Javad-Kalbasi, Student Member, IEEE, and Shahrokh Valaee,  Fellow, IEEE
	
	
	\thanks{Mohammed S. Al-Abiad,  Mohammad Javad-Kalbasi, and Shahrokh Valaee are with the Department of Electrical and Computer Engineering, University of Toronto, Toronto, Canada, 
Email: mohammed.saif@utoronto.ca, mohammad.javadkalbasi@mail.utoronto.ca, valaee@ece.utoronto.ca. 

A part of this paper \cite{saifglobecom} is accepted at the IEEE Globecom' 2023, Kuala Lumpur, Malaysia.

This work was supported in part by funding from the Innovation for Defence Excellence and Security (IDEaS) program from the Department of National Defence (DND).		
	}
	\vspace{-.35cm}
 }

\maketitle
\begin{abstract}
Reconfigurable intelligent surfaces (RISs) are expected to  make future 6G networks more connected and resilient  against node failures, due to their ability to introduce controllable phase-shifts onto impinging electromagnetic waves  and impose link redundancy. Meanwhile, unmanned aerial vehicles (UAVs) are prone to failure due to limited energy, random failures, or targeted failures, which causes network disintegration that results in information delivery loss. In this paper, we show that the integration between 
UAVs and RISs for improving network connectivity is crucial. We utilize RISs to provide path diversity and alternative connectivity options for information flow from user equipments (UEs) to less critical UAVs by adding more links to the network, thereby making the network more resilient and connected. To that end, we first define the criticality of UAV nodes, which reflects the importance of some nodes over other nodes. We then employ the algebraic connectivity metric, which is adjusted by the reflected links of the RISs and their criticality weights, to formulate the problem of maximizing the network connectivity. Such problem is a computationally expensive combinatorial optimization. To tackle this problem, we propose a relaxation method such that the discrete scheduling constraint of the problem is relaxed  and becomes continuous. Leveraging this, we propose two efficient solutions, namely semi-definite programming (SDP) optimization  and  perturbation heuristic, which both solve the problem in polynomial time. For the perturbation heuristic, we derive the lower and upper bounds of the algebraic connectivity obtained by adding new links to the network.  Finally, we corroborate the effectiveness of the proposed solutions through extensive simulation experiments.

\end{abstract}


\section{Introduction}
\subsection{Motivation}
In future 6G networks, there is a proliferation of connected nodes (i.e., smart devices, sensors, military vehicles) and services (i.e., augmented reality, information flow, data collection) that need to be supported by wireless networks \cite{1,2, MEC_saif}. Consequently, there is a pressing need for more connected wireless networks  that are resilient and robust  against node and link failures. To address this need, unmanned aerial vehicles (UAVs) communication is shown to be a promising solution since they can be rapidly deployed with adjustable mobility \cite{UAV_economy, saif}.  One distinctive aspect of UAV-assisted communications involves enhancing network connectivity through the establishment of line-of-sight (LoS) connections with user equipment (UE) \cite{8613833}. 

The prime concern of UAV communications is that UAVs\footnote{We use the terms nodes and UAVs/UEs interchangeably in this paper. In addition, the terms links and edges are used interchangeably since we usually use a graph to represent a network.}  are prone to failure due to several reasons, such as limited energy, hardware failure, or targeted failure in the case of battlefield surveillance systems. Such UAV failures cause  network disintegration, and consequently, information flow from UEs to a fusion center through UAVs can be severely impacted. Hence, it is crucial to always keep the  network connected, which can be addressed by adding more backhual links to the network \cite{H}. Network densification, by adding more UAVs or access points (APs), helps to improve network connectivity  but that will increase the hardware and energy consumption drastically \cite{M}. In addition, deploying a large number of UAVs or APs can be challenging in densely populated urban areas due to site constraints, limited space, and limited UAV battery. Instead, deploying passive nodes in the network can achieve a more connected system with less cost and lower energy consumption. 

Recently, reconfigurable intelligent surfaces (RISs) have drawn significant attention from both academia and industry, with the features of relatively low-cost and the capability to extend coverage and reduce energy consumption. Besides their other benefits in wireless networks, RISs offer two key advantages for designing connected and resilient wireless networks as follows. First, RISs can improve network connectivity by creating an indirect path between a UE and a UAV when the LoS is not available, or to provide more reliable links between the connected UEs and the UAVs \cite{M1, javadpaper}. As a result, path diversity and alternative connectivity options for information flow from UEs to UAVs are provided. Second, RIS provides a high passive beamforming gain via adjusting its reflection coefficients intelligently \cite{9293155}. By tuning the phase shifts of RIS, we can direct the signals of the UEs from  critical UAVs that are most likely to fail to less critical UAVs. Therefore, RIS can   reflect the signals of the UEs  to the desired UAVs, which is one of the main motivations of this work. As a result,  RISs can be utilized to provide for more connected and resilient networks that can effectively operate in the presence of node failures. In this paper, we show that with the integration of UAV communications and RISs, a small number of low-cost RISs can increase the connectivity of the RIS-assisted UAV networks significantly.

Optimized network connectivity plays a crucial role in  designing more connected and resilient networks and extending network lifetime, which is  defined  as the time until the network becomes disconnected \cite{4657335, 4786516}. An important  metric that measures how well a graph network is connected is called  the algebraic connectivity, also called the Fiedler metric or the second smallest eigenvalue\footnote{We use the terms algebraic connectivity and   the Fiedler value interchangeably in this paper since they both measure the network connectivity of the graph Laplacian matrix.} of the graph Laplacian matrix\cite{new}. One important property of the algebraic connectivity is that the larger it is, the more connected the graph will be. Such metric has been considered widely in wireless sensor networks. Several studies, as will be detailed below, consider routing solutions with the focus more on extending the battery lifetime of sensor nodes, e.g., \cite{4657335, 4786516, 1331424, 925682, 926982}. One trivial strategy to optimize network connectivity is to add more links via adding more connected nodes (i.e., APs, relays, or UAVs) \cite{4657335}. However, despite significant advancements in wireless sensor networks and UAV communications, the  limited battery of nodes, the consumed energy,  and cost necessitate a simple and efficient network design. Therefore, we introduce a cost-effective expansion of the traditional UAV communications by deploying compact RIS passive nodes that reflect the signals of the UEs to the desired UAVs. As a result, more links are added to the network to optimize its algebraic connectivity significantly.

\subsection{Related Work}
In the recent literature, several works have studied the importance of RISs in improving different metrices, such as positioning accuracy \cite{ M, M1}, extending the coverage of networks \cite{javadpaper}, boosting the communication capability \cite{9293155, RIS_annie, RIS_Mohanad, Javad_globecom2023}, etc. However, to the best of our knowledge, there is no work that used RISs to maximize network connectivity. Therefore, we briefly review some of the related works  that addressed  network connectivity maximization problem in different wireless sensor networks and UAV communications, e.g., \cite{4657335, 4786516, 1331424, 925682, 926982, 8292633, new_lifetime}. 

In spite of recent advances in wireless sensor networks, most of the existing studies consider routing solutions with the focus more on extending the battery lifetime of sensor nodes. These works define network connectivity as the network lifetime, in which the first node or all the nodes of a sensor network have failed \cite{4657335, 1331424, 926982}. In \cite{4657335}, the authors addressed the problem of adding relays to maximize the connectivity of multi-hop wireless networks. In \cite{8292633}, the authors maximized the algebraic connectivity by positioning the UAV to maximize the connectivity of small-cells backhaul network.  The paper \cite{11} proposed three different random relay deployment strategies, namely, the connectivity-oriented, the lifetime-oriented, and the hybrid-oriented strategies. However, there is no explicit optimization problem for maximizing the network lifetime in that work. A mathematical approach to placing a few flying UAVs over a wireless ad-hoc network in order to optimize the network connectivity for better coverage was proposed in \cite{12}.  However, none of the aforementioned works has ever explicitly considered the exploitation of RISs to add more reflected links to improve network connectivity of UAV-assisted networks. Different from the works \cite{4657335, 4786516, javadpaper, 1331424, 925682, 926982, 8292633, new_lifetime} that focused on routing solutions, this paper focuses on designing a more connected RIS-assisted UAV network. This network enables information flow from the UEs to the UAVs and is resistant to UAV failure.

The network connectivity maximization problem can be solved either by (i) relaxing the problem's constraints using convex relaxation, and then formulating the problem as a semi-definite programming (SDP) optimization problem to be solved using CVX \cite{saifglobecom, 4657335, 8292633} or (ii) using exhaustive search. However, the exhaustive search is computationally intractable for large network sizes and the SDP optimization is sub-optimal. Therefore, there is a need to find an efficient heuristic algorithm that can find a  set of suitable reflected links of the RISs that connect the UEs to the UAVs in polynomial time, such that we maximize the connectivity of a RIS-assisted UAV network. As one of the main contributions in this paper, we build on the reference \cite{book} to perturbate the eigenvalues of the  original Laplacian matrix with rank-one update matrix and propose a novel perturbation heuristic. This efficient perturbation heuristic is based on calculating the values of the Fiedler vector, which can be conveniently applied to large graphs with low computational complexity.

\subsection{Contribution}
In this paper, we investigate the integration between RISs and UAV-assisted communication systems by studying nodes criticality and the algebraic connectivity through SDP optimization and the Laplacian matrix perturbation, so that we maximize the connectivity of the envisioned RIS-assisted UAV network. The main contributions of this work are:
\begin{itemize}

\item \textbf{RIS-assisted UAV problem formulation:} First, we propose to define the criticality of UAV nodes, which reflects the importance of some nodes over other nodes towards the network connectivity. 
By considering the nodes’ importance from the graph connectivity perspective, the node with higher importance will be retained in the network, therefore the connectivity of the remaining network is maintained as long as possible. We then employ the algebraic connectivity metric, which is adjusted by the reflected links of the RISs and their criticality weights, to formulate the problem of maximizing the network connectivity. Such problem is shown to be a combinatorial optimization problem. By embedding the nodes criticality in the links selections, we propose two solutions to solve the proposed problem.

\item \textbf{Convex relaxation and SDP formulation:} To tackle this problem, we propose a relaxation method such that the objective function of the relaxed  problem becomes continuous. Leveraging this, we propose to formulate the problem as a SDP optimization problem that can be solved efficiently in polynomial time using CVX.

\item \textbf{Algebraic connectivity perturbation:} 
We propose a low-complexity, yet efficient, perturbation heuristic, which has less complexity compared to the relaxation method. In this heuristic perturbation, one UE-RIS-UAV link is added at a time by calculating only the  eigenvector values corresponding to the algebraic connectivity. We also derive the lower and upper bounds of the algebraic connectivity obtained by adding new links to the network based on this perturbation heuristic.

\item \textbf{Performance evaluation:} We evaluate the performance of the proposed schemes in terms of network connectivity via extensive simulations. We verify that the proposed schemes result in improved  network connectivity as compared to the existing solutions. In particular, the proposed perturbation heuristic has a superior performance that is roughly the same as the optimal solution using exhaustive search and close to the upper bound. 

\end{itemize}

The rest of this paper is organized as follows. In \sref{S}, we describe the system model, network connectivity, and then define nodes criticality   and outline some of its important properties. In \sref{PF}, the network connectivity maximization problem in a RIS-assisted UAV network is formulated. In \sref{PS}, the proposed solutions are explained, and the upper and lower bounds on the algebraic connectivity of our proposed perturbation scheme are analyzed in \sref{LU}. Extensive simulation results  are presented in \sref{NR}, and the conclusion is given in \sref{C}.

\begin{figure}[t!]

	\centerline{\includegraphics[width=0.85\linewidth]{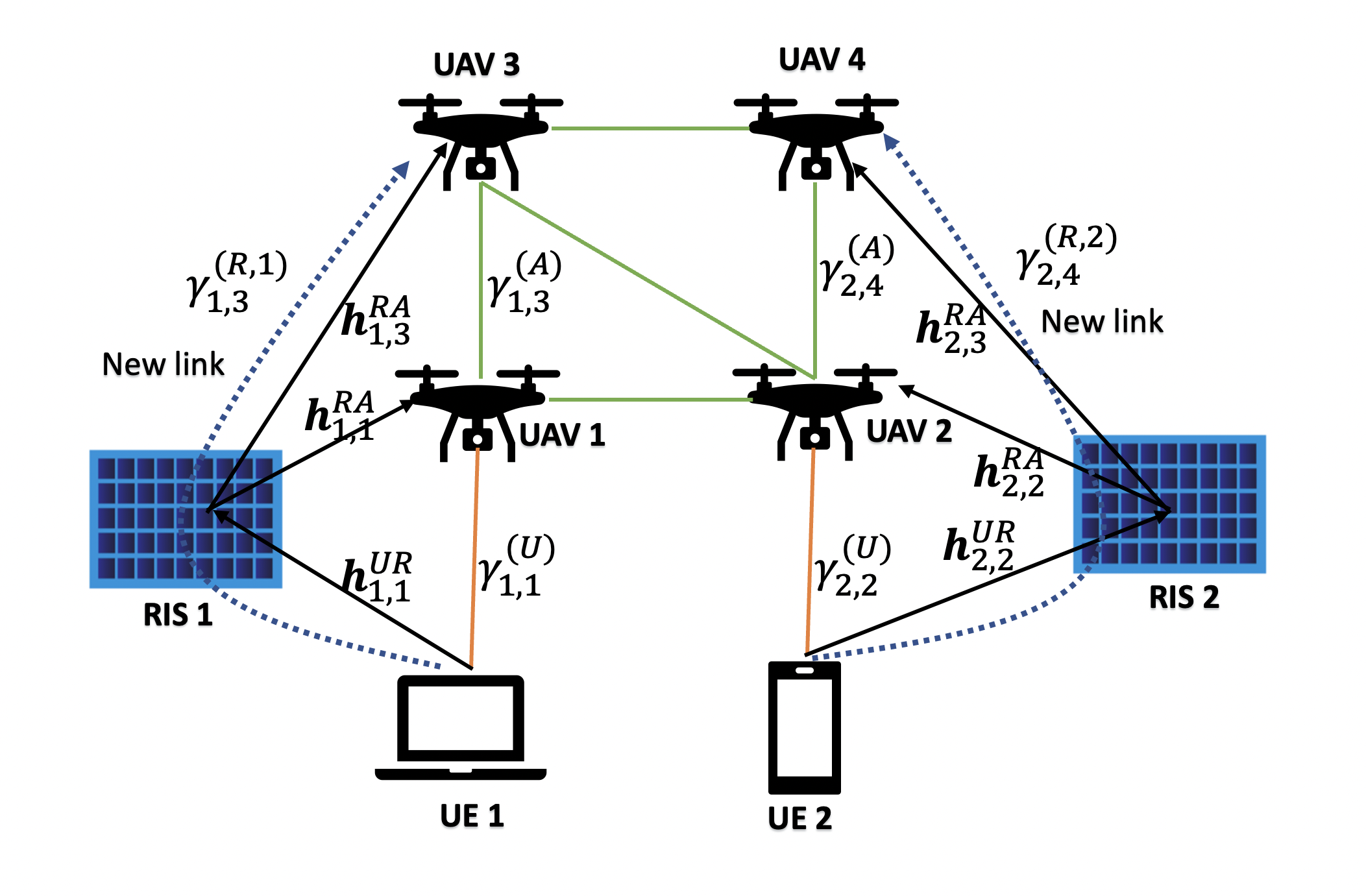}}
	
	\caption{A typical RIS-assisted UAV network with $2$ RISs, $2$ UEs, and $4$ UAVs.}
	\label{fig1}
\end{figure}

\section{System Model and Network Connectivity}\label{S}

\subsection{System Model}\label{SA}
Consider the sample RIS-assisted UAV system shown in \Fref{fig1}. The system consists of a set of UAVs, a set of RISs, and multiple UEs that represent ground users, sensors, etc. The sets of UAVs, RISs, and UEs are denoted as $\mathcal A=\{1, 2, \ldots, A\}$, $\mathcal R=\{1, 2, \ldots, R\}$, and $\mathcal U=\{1, 2, \ldots, U\}$, respectively. All UEs and UAVs are equipped with single antennas. 
The $A$ UAVs fly and hover over
assigned locations at a fixed flying altitude and connect with $U$ UEs. The locations of the UAVs, UEs, and RISs are assumed to be fixed. The RISs are installed with a certain altitude $z_r$, $r \in \mathcal R$. Let $(x_r, y_r, z_r)$ be the 3D location of the $r$-th RIS, $(x_a, y_a, z_a)$  the 3D location of the $a$-th UAV, 
and $(x_u, y_u)$  the 2D location of the $u$-th UE, respectively. The distances between the $u$-th UE and the $r$-th RIS and between the $r$-th RIS and the $a$-th UAV are denoted by $R^\text{UR}_{u,r}$ and $R^\text{RA}_{r,a}$, respectively.

Due to their altitude, UAVs can have good  connectivity to UEs. However, UEs may occasionally experience deep fade. To overcome this problem and further improve network connectivity, we propose to utilize a set of RISs to impose link redundancy to the network. As such, the network becomes more resilient against node failures by providing path diversity and alternative connectivity options between the UEs and the UAVs. Each RIS is equipped with a controller and $M:=M_b \times M_c$ passive reflecting units (PRUs) to form a uniform passive array (UPA). Each column of the UPA has $M_b$ PRUs with an equal spacing of $d_c$ meters (m) and each row of the UPA consists of $M_c$ PRUs with an equal spacing of $d_b$ m. Through appropriate adjustable phase shifts, these PRUs can add indirect links between the UEs and the UAVs or connect the blocked UEs to the desired UAVs. The phase-shift matrix of the $r$-th RIS is modeled as the diagonal matrix $\mathbf{\Theta}_r =diag(e^{j\theta_{1}^r},\ldots, e^{j\theta_{2}^r},\ldots,e^{j\theta_{M}^r})$, where $\theta_{m}^{r}\in [0,2\pi)$, for $r\in \mathcal R$, and $m\in \{1, \ldots, M\}$.

For the $u$-th UE, the reachable UAVs are denoted by a set $\mathcal A^u$. Assuming that $|\mathcal A^u|\geq 1$ in general, some UEs are able to access multiple UAVs simultaneously. Similarly, the reachable RISs for  the $u$-th UE are denoted by a set $\mathcal R_u$. The successful communications between the UEs and RISs are measured using the distance threshold $D_o$, i.e., the $u$-th UE can be connected to the $r$-th RIS with distance $d_{u,r}$ if $d_{u,r} \leq D_o$.
The communications between the UEs and UAVs/RISs are assumed to occur over different time slots (i.e., time multiplexing access) to avoid interference among the scheduled UEs. Therefore, in each time slot, we assume that distinct UEs $u$ and $u'$ are allowed to transmit simultaneously if 
$\mathcal R_u \cap \mathcal R_{u'} =\emptyset$ to reduce interference.

We focus on the network connectivity from data link-layer viewpoint, thus we abstract the physical layer factors and consider a model that relies only on the distance between the nodes. Therefore, similar to \cite{8292633}, we model only the large scale fading and ignore the small scale fading. To quantify the UEs transmission to the UAVs/RISs, we use the signal-to-noise ratio (SNR). For the $u$-th UE, SNR is defined as follows \cite{8292633}
\begin{equation}
\gamma^{(U)}_{u,a}=\frac{d_{u,a}^{-\alpha}p}{N_0},
\end{equation}
where $d_{u,a}$ is the distance between the $u$-th UE and the $a$-th UAV, $p$ is the transmit power of the $u$-th UE, which is maintained fixed for all the UEs, $N_0$ is the additive white Gaussian noise (AWGN) variance, and $\alpha$ is the path loss exponent that depends on the transmission environment.

UAVs hover at high altitudes, hence we reasonably  assume that they maintain LoS channel between each other \cite{8613833}. The path loss between the $a$-th UAV and $a'$-th UAV can be expressed as
\begin{equation}
\Gamma_{a,a'}=20\log\bigg( \frac{4 \pi f_c d_{a,a'}}{c} \bigg),
\end{equation}
where $d_{a,a'}$ is the distance between the $a$-th UAV and the $a'$-th UAV, $f_c$ is the carrier frequency, and $c$ is the speed of light. Consequently, the SNR between the $a$-th and the $a'$-th UAVs is
\begin{equation}
\gamma^{(A)}_{a,a'}=10\log P-\Gamma_{a,a'}-10 \log N_0, ~~~~ (dB)
\end{equation}
where $P$ is the transmit power of the $a$-th UAV, which is maintained fixed for all the UAVs. Note that the SNR of the $u$-th UE determines whether it has a successful connection to the corresponding UAV $a$. In other words, the $a$-th UAV is assumed to be within the transmission range of the $u$-th UE if $\gamma^{(U)}_{u,a} \geq \gamma^\text{UE}_{0}$, where $\gamma^\text{UE}_{0}$ is the minimum SNR threshold for the communication links between the UEs and the UAVs.  Similarly, we assume that UAV $a$ and UAV $a'$ have a successful connection provided that $\gamma^{(A)}_{a,a'} \geq \gamma^\text{UAV}_{0}$, where $\gamma^{(UAV)}_{0}$ is the minimum SNR threshold for the communication links between the UAVs. 

Since RISs are deployed in high altitude, the signal propagation of UE-RIS link is adopted to be a simple yet reasonably accurate LoS channel model \cite{9293155}. The LoS channel vector between the $u$-th UE and the $r$-th RIS is given by \cite{9293155}
\begin{equation}
\mathbf h^\text{UR}_{u,r}=\sqrt{\frac{\beta_0}{(d_{u,r}^\text{UR})^2}}\tilde{\mathbf h}^\text{UR}_{u,r},
\end{equation}
where $d_{u,r}^\text{UR}$ is the distance between the $u$-th UE and the $r$-th RIS, $\beta_0$ denotes the path loss at the reference distance $d=1$m, and $\tilde{\mathbf h}^\text{UR}_{u,r}$ 
represents the array response component, which can be denoted by
\begin{eqnarray}\nonumber
\tilde{\mathbf{h}}^\text{UR}_{u,r}&\!\!=\!\!&{[1,e^{-j\frac{2\pi d_{b}}{\lambda}\phi_{u,r}^\text{UR}\psi_{u,r}^\text{UR} },\ldots,e^{-j\frac{2\pi d_{b}}{\lambda}(M_{b}-1)\phi_{u,r}^\text{UR}\psi_{u,r}^\text{UR} }]}^{T}\\\nonumber
&~&\otimes {[1,e^{-j\frac{2\pi d_{c}}{\lambda}\varphi_{u,r}^\text{UR}\psi_{u,r}^\text{UR} },\ldots,e^{-j\frac{2\pi d_{c}}{\lambda}(M_{c}-1)\varphi_{u,r}^\text{UR}\psi_{u,r}^\text{UR} }]}^{T},
\end{eqnarray}
where $\lambda$ is the wavelength, 
$\otimes$ is the Kronecker product, and 
$\phi_{u,r}^\text{UR}, \varphi_{u,r}^\text{UR}$, and $\psi_{u,r}^\text{UR}$ are related to the sine and cosine terms of the vertical and horizontal angles-of-arrival (AoAs)
at the $r$-th RIS \cite{9293155}, and respectively given by $\phi_{u,r}^\text{UR}=\frac{y_{u}-y_{r}}{\sqrt{(x_u-x_r)^2+(y_u-y_r)^2}}$, $\varphi_{u,r}^\text{UR}=\frac{x_{r}-x_{u}}{\sqrt{(x_u-x_r)^2+(y_u-y_r)^2}}$, $\psi_{u,r}^\text{UR}=\frac{-z_{r}}{d^\text{UR}_{u,r}}$. On the other hand, the RISs and UAVs are deployed in high altitudes, thus the reflected signal propagation of the RIS-UAV link typically occurs in clear airspace where the obstruction or reflection effects diminish. The LoS channel vector between the $r$-th RIS and the $a$-th UAV is given by
\begin{equation}
\mathbf h^{RA}_{r,a}=\sqrt{\frac{\beta_0}{(d_{r,a}^{RA})^2}}\tilde{\mathbf h}^{RA}_{r,a},
\end{equation}
where $d_{r,a}^\text{RA}$ is the distance between the $r$-th RIS and the $a$-th UAV, and $\tilde{\mathbf h}^\text{RA}_{r,a}$ 
represents the array response component which can be denoted by
\begin{eqnarray}\nonumber
\tilde{\bold{h}}^\text{RA}_{r,a}\!\!&=\!\!&{[1,e^{-j\frac{2\pi d_{b}}{\lambda}\phi_{r,a}^\text{RA}\psi_{r,a}^\text{RA} },\ldots,e^{-j\frac{2\pi d_{b}}{\lambda}(M_{b}-1)\phi_{r,a}^\text{RA}\psi_{r,a}^\text{RA} }]}^{T}\\\nonumber
&~&\otimes {[1,e^{-j\frac{2\pi d_{c}}{\lambda}\varphi_{r,a}^\text{RA}\psi_{r,a}^\text{RA} },\ldots,e^{-j\frac{2\pi d_{c}}{\lambda}(M_{c}-1)\varphi_{r,a}^\text{RA}\psi_{r,a}^\text{RA} }]}^{T},
\end{eqnarray}
where $\phi_{r,a}^\text{RA}, \varphi_{r,a}^\text{RA}$, and $\psi_{r,a}^\text{RA}$ are related to the sine and cosine terms of the vertical and horizontal angles-of-departure (AoDs)
from the $r$-th RIS to the $a$-th UAV \cite{9293155}, and respectively given by $\phi_{r,a}^\text{RA}=\frac{y_{r}-y_{a}}{\sqrt{(x_r-x_a)^2+(y_r-y_a)^2}}$, $\varphi_{r,a}^\text{RA}=\frac{x_{r}-x_{a}}{\sqrt{(x_r-x_a)^2+(y_r-y_a)^2}}$, and $\psi_{r,a}^\text{RA}=~\frac{z_r-z_a}{d^\text{RA}_{r,a}}$.

Based on the channel models described above, the concatenated channel for the
UE-RIS-UAV link between the $u$-th UE and the $a$-th UAV through the $r$-th RIS is given by \cite{9293155}
\begin{equation}\label{6line}
\mathbf h^\text{URA}_{u,a} =  (\mathbf h^\text{RA}_{r,a} \mathbf)^H \mathbf{\Theta}_r \mathbf h^\text{UR}_{u,r}. 
\end{equation}
Accordingly, the SNR of the reflected link between the $u$-th UE and the $a$-th UAV through the $r$-th RIS can be written as \cite{9593204} 
\begin{equation}\label{7line}
\gamma^{(R,r)}_{u,a}= \frac{p \lVert \mathbf h^\text{URA}_{u,a}\rVert^2}{N_0}.
\end{equation}

In this paper, we model the considered RIS-assisted UAV network as an undirected
 graph $\mathcal G(\mathcal V, \mathcal E)$, where $\mathcal V = \{v_1, v_2, \cdots, v_V\}$ is the
set of nodes (i.e., UAVs and UEs) in the network, and $\mathcal E= \{e_1, e_2, \cdots, e_E\}$ is the set of all edges, where $V$ and $E$ are the cardinality of the sets $\mathcal V$ and $\mathcal E$, respectively, i.e.,  
$V=|\mathcal U \cup \mathcal A| =|\mathcal V|$ and $E= |\mathcal E|$. 
The edge between any two nodes is created based on a typical previously mentioned SNR threshold. The key notations are summarized in \tref{table_new}.

	\begin{table*}[t!]
	\renewcommand{\arraystretch}{0.9}
	\caption{Summary of the Main Notations and Definitions}
	\label{table_new}
	\centering
	\begin{tabular}{|p{2.3cm}| p{11.0cm}|}
		\hline
		\textbf{Variable} & \textbf{Definition}\\
		\hline
		\hline
		$\mathcal{A}, \mathcal U$, $\mathcal R$ & Sets of $A$ UAVs, $U$ UEs, and $R$ RISs \\
  	\hline
		$\mathcal U_t$ & Set of transmitting UEs in time slot $t$\\
		\hline
		$R^\text{UR}_{u,r}, R^\text{UR}_{r,a}$ & Distances between UE $u$ (RIS $r$) and RIS $r$ (UAV $a$) \\
         \hline
          $\theta^r_m$ & Phase-shift of element $m$ at RIS $r$\\
		\hline
          $d_c, d_b$ & Column and row spacing between RIS elements\\
           \hline
		$\mathbf{\Theta}_r, M$ & Phase-shift matrix of RIS $r$ and total number of RIS elements\\
		\hline
		$\mathcal R_u, D_0$ & Reachable RISs to UE $u$ and distance threshold\\
		\hline
		$\gamma^{(U)}_{u,a}, \gamma^{(A)}_{a,a'}$ & SNRs of UE $u$ to UAV $a$ and UAV $a$ to UAV $a'$\\
		\hline
		$\gamma^\text{(UE)}, \gamma^\text{(UAV)}$ &  SNR thresholds for UE-UAV and UAV-UAV communications \\
		\hline
      $\mathbf h^\text{UR}_{u,r}, \mathbf h^\text{RA}_{r,a}$ &  LoS channel vectors between UE $u$ (RIS $r$) and RIS $r$ (UAV $a$)\\
		\hline
		$\gamma^{(R,r)}_{u,a}$ & SNR of the reflected link between UE $u$ and UAV $a$ through RIS $r$\\
  	\hline
		$\mathcal G, \mathcal G_{-n}$ & Original graph and its sub-graph after removing node $n$ and all its edges\\
		\hline
		$\mathcal V, \mathcal E$ & Sets of vertices and edges of graph $\mathcal G$\\
				\hline
		$\mathbf L, \mathbf A$ & Laplacian and incidence matrices\\
  	    \hline
		$deg(v_{n})$ & Degree of node $v_n$ \\
  	\hline
		$w_l$ & Weight of link $l$ that connects UE $u$ and UAV $a$ through a RIS\\
		\hline
		$\lambda_2(\mathbf L)$ & Algebraic connectivity of  the original Laplacian matrix $\mathbf L$\\
		\hline
  $\lambda_2(\mathbf L')$ & Algebraic connectivity of  adjusted Laplacian matrix $\mathbf L'$\\
		\hline
		$\mathcal C_n$ & Criticality of node $n$\\
		\hline
		$X_{u,a}, Y^u_{r,a}$ & Two binary variables for UE-UAV and UE-RIS-UAV scheduling\\
		\hline
			\end{tabular}
\end{table*}

\subsection{Network Connectivity}\label{NC}
This subsection briefly discusses the definition of the Laplacian matrix $\mathbf L$ representing a graph $\mathcal G(\mathcal V, \mathcal E)$, its second smallest eigenvalue $\lambda_2(\mathbf L)$, and the relationship between $\lambda_2(\mathbf L)$ and the connectivity of the associated graph $\mathcal G(\mathcal V, \mathcal E)$.

For an edge $e_{k}$, $ 1 \leq k \leq E$, that connects two nodes $\{v_n, v_m\} \in \mathcal V$, let $\mathbf a_k$ be a vector, where the $n$-th and $m$-th elements in $\mathbf a_k$ are given by $a_{k,n}=1$ and $a_{k,m}=-1$, respectively, and zero otherwise. The incidence matrix $\mathbf A$ of a graph $\mathcal G$ is the matrix with the $k$-th column given by $\mathbf a_k$.  Furthermore, the weight of an edge $e_k$ that connects two nodes $\{v_n, v_m\} \in \mathcal V$, denoted by $w_{n,m}$ or $w_k$, is a function of the criticality of the nodes, as will be discussed in \sref{nc}. The weight vector $\mathbf w \in {[\mathbb R^+]}^E$ is defined as $\mathbf w= [ w_1, w_2, \ldots, w_E]$. Hence, in an undirected graph $\mathcal G(\mathcal V, \mathcal E)$, the Laplacian matrix $\mathbf L$ is a $V$ by $V$ matrix, which is defined as follows \cite{4657335}:
\begin{equation} \label{lap}
\mathbf L= \mathbf A  ~diag(\mathbf w) ~\mathbf A^T=\sum^{E}_{k=1} w_k \mathbf a_k \mathbf a^T_k,
\end{equation}
where the entries of $\mathbf L$ are given element-wise by
\begin{equation}
L(n,m) = \begin{cases}
deg(v_{n}) &\text{if} ~v_n=v_m,\\
-w_{n,m} &\text{if}~ (v_n, v_m) \in \mathcal E, \\
0 & \text{otherwise},
\end{cases}
\end{equation}
where $n, m \in \{1, 2, \ldots, V\}$ are the indices of the nodes, and $deg(v_{n})$ is the degree of node $v_n$, which represents the number of all of its neighboring nodes. 

As mentioned above, in network connectivity, algebraic connectivity  measures how well a graph $\mathcal G$ that has the associated Laplacian matrix $\mathbf L$ is connected \cite{new}. This metric is usually denoted as $\lambda_2(\mathbf L)$. 
The main motivation of $\lambda_2(\mathbf L)$  as a network connectivity metric comes from the following two main reasons \cite{new}. First, $\lambda_2 (\mathbf L) > 0$ if and only if $\mathcal G$ is connected, i.e., $\mathcal G$ has only one component. It is worth mentioning that when $\lambda_2(\mathbf L)=0$, the graph is disconnected in which it has more than one component. Second,  $\lambda_2 (\mathbf L)$ is monotone increasing in the edge set, i.e., if $\mathcal G_1=(\mathcal V, \mathcal E_1)$ is associated with $\mathbf L_1$ and $\mathcal G_2=(\mathcal V, \mathcal E_2)$ is associated with $\mathbf L_2$ and $\mathcal E_1 \subseteq \mathcal E_2$, then $\lambda_2 (\mathbf L_1) \leq \lambda_2 (\mathbf L_2)$. This implies that $\lambda_2(\mathbf L)$ qualitatively represents the connectivity of a graph in the sense that the larger $\lambda_2(\mathbf L)$ is, the more connected the graph will be. To this end, since $\lambda_2 (\mathbf L)$ is a good measure of how connected the graph is, the more carefuly selected edges that exist between the UEs and the UAVs, the longer the network can live without being disconnected due to node failures. Thus, the  network becomes more  {\it resilient}. Based on that and similar to \cite{4657335, new_lifetime},  this paper  considers $\lambda_2 (\mathbf L)$ as a quantitative measure of network connectivity, which shows network resiliency against node failures.  For simplicity, the remaining of the paper uses node $n$ instead of node $v_n$.

\subsection{Nodes Criticality}\label{nc}
Let $\mathcal G_{-n}$ be the remaining graph after removing UAV node $n$ and all its adjacent edges to other nodes. Noticed that the most critical nodes in $\mathcal G$ are those representing the UAVs since UAVs have many connections to UEs and other UAVs, i.e., they represent the backhaul core of the network.  We propose to quantify the connectivity of the remaining graph based on the Fiedler value. The connectivity of the remaining graph can be quantified by $\lambda_2(\mathcal G_{-n})$ of the graph resulted from removing a typical node $n$ and all its connected edges in $\mathcal G$. We define the nodes criticality as follows.
\begin{definition} \label{def1}
The criticality of node $n$, which reflects the severity of network connectivity after removing that node and its connected edges to other nodes, is defined as 
\begin{equation}\label{eq1}
\mathcal C_n=\frac{1}{\lambda_2(\mathcal G_{-n})}.
\end{equation}
\end{definition}
\qed


\textbf{Remark:} 
Since  the Laplacian matrix $\mathbf L$ is positive semi-definite (expressed as $ 0 \preceq \mathbf L$), we have 
$\mathcal C_n \geq 0$.

\begin{theorem} $ \frac{1}{\lambda_2(\mathcal G)-1}$ is a tight upper bound on $\mathcal C_n $.
\end{theorem}

\begin{proof}  Consider the graph $\mathcal G (\mathcal V, \mathcal E)$ with the set of vertices, $\mathcal V$, and the set of edges, $\mathcal E$. Let us define a new graph $\mathcal G_{com}(\mathcal V, \mathcal E_{com})$ by extending $\mathcal G (\mathcal V, \mathcal E)$ with adding all missing edges from node $n$, such that, $\mathcal E \subseteq \mathcal E_{com}$, and node $n$ is connected to all other nodes in the graph. Then, $\mathbf L(\mathcal G_{com})$ can be written as 
\begin{equation}\label{L_comp}
\mathbf L(\mathcal G_{com}) = 
\begin{pmatrix}
\mathbf L(\mathcal G_{-n})+\mathbf I,  & -\mathbf 1 \\
-\mathbf 1^T, & V-1
\end{pmatrix}.
\end{equation}
Let $\mathbf v$ be an eigenvector of $\mathbf L(\mathcal G_{-n})$ that is corresponding to $\lambda_2(\mathbf L(\mathcal G_{-n}))$. From \eref{L_comp}, we show that $\lambda_2(\mathcal G_{-n})+1$ is an eigenvalue of $\mathbf L(\mathcal G_{com})$ as follows: 
\begin{align*}   
\mathbf L(\mathcal G_{com}) 
\begin{pmatrix}
\mathbf v\\
0
\end{pmatrix} & = 
\begin{pmatrix}
\mathbf L(\mathcal G_{-n})+\mathbf I,  & -\mathbf 1 \\
-\mathbf 1^T, & V-1
\end{pmatrix}
\begin{pmatrix}
\mathbf v\\
0
\end{pmatrix} \\& 
= \begin{pmatrix}
\mathbf L(\mathcal G_{-n})\mathbf v +\mathbf I \mathbf v\\
- \mathbf 1^T \mathbf v
\end{pmatrix}
\stackrel{({\rm a})}{=} \begin{pmatrix}
 \lambda_2(\mathcal G_{-n})\mathbf v+\mathbf v\\
0
\end{pmatrix} \\ &
= \begin{pmatrix}
 (\lambda_2(\mathcal G_{-n})+1 )\mathbf v\\
0
\end{pmatrix}=
[\lambda_2(\mathcal G_{-n})+1]
\begin{pmatrix}
 \mathbf v\\
0
\end{pmatrix},
\end{align*}
where (a) follows from the fact that $- \mathbf 1^T \mathbf{v}=0$ for a connected graph, and $\mathbf I \mathbf v=\mathbf v$. Thus, 
$\lambda_2(\mathcal G_{-n})+1$ is an eigenvalue of $\mathbf L(\mathcal G_{com})$ that is different from the zero eigenvalue, i.e., $\lambda_2(\mathcal G_{-n})+1 \neq \lambda_1(\mathcal G_{com})$ and $\lambda_2(\mathcal G_{com}) \leq \lambda_2(\mathcal G_{-n})+1$.
Therefore, 
\begin{align} \label{444}
\lambda_2(\mathcal G_{-n}) \geq  \lambda_2(\mathcal G_{com})-1.
\end{align}
Finally, since $\mathcal E \subseteq \mathcal E_{com}$,  $\lambda_2(\mathcal G) \leq \lambda_2(\mathcal G_{com})$, (\ref{444}) can be written as 
\begin{align} \label{555}
\lambda_2(\mathcal G_{-n}) \geq  \lambda_2(\mathcal G)-1.
\end{align}
This shows that removing node $n$ from $\mathcal G$ can at least reduce the algebraic connectivity  $\lambda_2(\mathcal G)$ by $1$, which depends on the number of edges that connect node $n$ to the remaining nodes in $\mathcal G$. From (\ref{555}) and (\ref{eq1}), we have 
\begin{align}  \label{666}
\mathcal C_n \leq  \frac{1}{\lambda_2(\mathcal G)-1}.
\end{align}
Inequality (\ref{666}) reasonably implies that if $\mathcal G$ is highly connected, removing a node from it would not affect the network connectivity much since all nodes have high criticality and more connected. 
Assuming that $\mathcal G$ is a complete graph, $\lambda_2(\mathcal G) =V$, and accordingly, $\lambda_2(\mathcal G_{-n}) =\lambda_2(\mathcal G)-1=V-1$. In this case, we have
\begin{equation} \nonumber 
\mathcal C_n = \frac{1} {\lambda_2(\mathcal G_{-n})}=  \frac{1}{V-1},
\end{equation}
Which shows the tightness of upper bound. We point out that $\mathcal C_n \approx  \frac{1}{V}$ for large complete graphs.
\end{proof}

\dref{def1} implies that the high criticality of nodes  that cause severe reduction in the remaining algebraic network connectivity should be balanced. In this case, no node in the graph can severely impact the network resiliency if it has accidentally or intentionally failed. Intuitively to make the network more connected, the RISs should be utilized properly since they have less probability of failures as compared to UAVs due to their limited batteries. In addition, RISs will add more links to the network, which will maximize the network connectivity and reduce the criticalities of the nodes. Consequently, the signals of the UEs can be redirected via the RISs to less critical UAV nodes, which results in more balanced network. One can achieve this balance by assigning weights to edges connecting UE nodes to UAV nodes, and the selection of the UE-RIS-UAV combinations relatively rely  on these weights. Thus, we propose to design the weight of link $e_k$ that connects nodes $n$ and $m$ as follows 
\begin{equation}\label{wei}
w_{n,m}=\frac{1}{\mathcal C_n+\mathcal C_m},
\end{equation}
where $\mathcal C_n$ and $\mathcal C_m$ are the criticalities of nodes $n$ and $m$, respectively. The weight $w_{n,m}$ is high if the criticalities of the nodes $n$ and $m$ are low, thus a link between them via a possible RIS would be highly created.  The high critical UAV nodes are less likely to have new links from the UEs via the RISs since their failures may significantly degrade the network connectivity. Therefore, by carefully adding new links to less critical UAV nodes, one can balance the criticality of all the UAV nodes in the network.

\section{Problem Formulation}\label{PF}
 
The problem of  maximizing the connectivity  of RIS-assisted UAV networks can be stated as follows. 
Given a RIS-assisted UAV network represented by a graph $\mathcal G(\mathcal V, \mathcal E)$, what are the optimum combinations between the UEs and the UAVs via optimizing phase shifts of the RISs in order to maximize $\lambda_2 (\mathbf L)$ of the resulting network while balancing the criticality of the UAV nodes? 
Essentially, deploying the RISs in the network may result in connecting multiple UEs to multiple UAVs, which were not connected together. It may also result in adding new alternative options to the UEs if their scheduled high critical UAV nodes have failed. By adjusting their phase shifts, RISs can smartly beamform the signals of the UEs to the desired, less critical UAVs to maximize the network connectivity and make the network more {\it resilient} against node failures. 

We consider that multiple UEs are allowed to transmit simultaneously if their mutual transmission coverage to the RISs is empty.  Therefore, we have the following UE-RIS-UAV association constraints:
\begin{itemize}
\item Multiple UEs are allowed to transmit as long as they do not have common RISs in their coverage transmissions.

\item Each RIS is connected to only one UE. 
\item Each RIS reflects the signal of the selected UE to only one UAV.
\item Each UAV is connected to one RIS only. 
\end{itemize}

Let $\mathcal A^{u,r}_0$ be a set of reachable UAVs that have indirect communication links from the $u$-th UE through the $r$-th RIS, i.e., $\mathcal A^{u,r}_0 =\{a\in \mathcal A\backslash \mathcal A^{u} \mid \gamma^{r}_{u,a} \geq \gamma^\text{RIS}_{0}\}$, where $\mathcal A^{u}$ is defined above as the set of UAVs that have direct links to the $u$-th UE. We aim at providing alternative links to connect the UEs to the suitable UAVs in the set $\mathcal A_0^{u,r}$, $\forall u \in \mathcal U, r \in \mathcal R$. As such, the UEs do not miss the communications if their scheduled UAV have failed.
Let  $X_{u,r}$ be a binary variable that is equal to $1$ if the $u$-th UE is connected to the $r$-th RIS, and zero otherwise. Now, let  $Y^u_{r,a}$ be a binary variable that is equal to $1$ if the $r$-th RIS is connected to the $a$-th UAV when the $u$-th UE is selected to transmit, and zero otherwise. Let $\mathcal U_t$ be the set of the possible transmitting UEs, where $\mathcal U_t \subset \mathcal U$. This set $\mathcal U_t$ consists of multiple UEs that have empty mutual transmission coverage to the RISs, which can be defined as $\mathcal U_t=\{u \in \mathcal U \mid (u, u') \in \mathcal U^2, \mathcal R_u \cap \mathcal R_{u'}= \emptyset\}$. Therefore, the considered optimization problem of  maximizing the network connectivity  is formulated as follows:
\begin{subequations}
\label{eqob}
\begin{align}
&\max_{\substack{X_{u,r}, Y^u_{r,a}, \theta^r_m}}  \lambda_2(\mathbf L')
\label{eqoba} \\
 &{\rm subject~to\ }  \sum\limits_{u \in \mathcal U_t} X_{u,r}=1, ~~~~~~~~~~~~~~~~~~~~\forall r \in \mathcal R,  \label{eqobb} \\
 & ~~~~~~\sum_{u \in \mathcal U_t} \sum_{a \in \mathcal A^{u,r}_0} Y^u_{r,a} \leq 1, ~~~~~~~~~~~~~~~~~~~~~\forall r \in \mathcal R,  \label{eqobc}\\
& ~~~~~~\theta^r_m \in [0, 2\pi),~~~~~~~~ \forall r \in \mathcal R, m=\{1, \ldots, M\}, \label{eqobf} \\
& ~~~~~~X_{u,r}\in \{0,1\}, Y^u_{r,a}\in \{0,1\},~~~\forall r \in \mathcal R, a \in \mathcal A.
\end{align}
\end{subequations} 
In (\ref{eqob}), constraint (\ref{eqobb}) implies that each RIS receives a signal from a single UE in the set $\mathcal U_t$. 
Constraint (\ref{eqobc}) assures that each RIS reflects the signal of its associated UE  to only one UAV.  This also means  that at maximum $R$ paths from the selected UEs to the suitable UAVs through the RISs. 
Constraint (\ref{eqobf}) is for the RIS phase shift optimization.

Since the above optimization problem  (\ref{eqob}) is NP-hard, we first propose heuristic solutions to find feasible UE-RIS-UAV associations. Afterwards, we optimize the phase shift of the RISs elements to smartly direct the signals of the UEs to the suitable UAV nodes.

\section{Proposed Solutions}\label{PS}
In this section, we  reformulate the problem \eref{eqob}, and then develop two heuristic solutions to maximize $\lambda_2 (\mathbf L')$. In \sref{CR}, we relax the problem in \eref{eq10a} to a convex optimization problem in order to be formulated as an SDP problem. In \sref{PH}, we propose a novel perturbation heuristic that selects the maximum possible increase in $\lambda_2 (\mathbf L')$ based on the weighted values of the differences between the values of the Fiedler vector.

We add a link connecting the $u$-th UE to the $a$-th UAV through the $r$-th RIS if both $X_{u,r}$ and $Y^u_{r,a}$ in (\ref{eqob}) are $1$. Let $\mathbf z$ be a vector representing the UE-RIS-UAV candidate associations, in which case $X_{u,r}=1$ and $Y^u_{r,a}=1$, $\forall u \in \mathcal U, a \in \mathcal A, r \in \mathcal R$. Therefore, the problem in (\ref{eqob}) can be seen as having a set of $|\mathbf z|$ UE-RIS-UAV candidate associations, and we want to select the optimum $R$ UE-RIS-UAV associations among these
$|\mathbf z|$ associations. This optimization problem can be formulated as
\begin{align} 
&\max_{\mathbf z} ~~~~~~~~~~\lambda_2(\mathbf L'(\mathbf z))
\label{eq10a} \\
& {\rm subject~to\ } ~~~~\mathbf 1^T \mathbf z =R, \nonumber \\
& ~~~~~~~~~~~~~~~~~~\mathbf z \in \{0,1\}^{|\mathbf z|}, \nonumber
\end{align}
where $\mathbf 1 \in \mathbf R^{|\mathbf z|}$ is the all-ones vector and
\begin{equation}\label{update}
\mathbf L'(\mathbf z)=\mathbf L+\sum^{|\mathbf z|}_{l=1} z_l w_l \mathbf a_l \mathbf a^T_l,
\end{equation}
where $\mathbf a_l$ is the incidence vector resulting from adding link $l$ to the original graph $\mathcal G$, $\mathbf 1^T \mathbf z=R$ indicates that the number of chosen RISs is $R$, $\mathbf L$ is the Laplacian matrix of the original graph $\mathcal G$, and $w_l$ is the weight of a constructed edge that connects a RIS with a UE and a UAV as given in \eref{wei}. The $l$-th element of $\mathbf z$, denoted by $z_l$, is either
$1$ or $0$, which corresponds to whether a RIS should be
chosen or not, respectively.  Clearly, the dimension of $\mathbf L$ and $\mathbf L'(\mathbf z)$ is $V \times V$.  We notice that the effect of adding RISs appears only in the edge set $\mathcal E$, and not in the node set $\mathcal V$. For ease of illustration, \eref{update} can be written as 
\begin{equation} 
\mathbf L'=\mathbf L+(\mathbf z \otimes \mathbf I_{|\mathbf z|})\Lambda,
\end{equation}
where $\Lambda = \big [ (\mathbf A_1 diag(\mathbf w_1)\mathbf A^T_1)^T, \ldots,  (\mathbf A_{|\mathbf z|} diag(\mathbf w_{|\mathbf z|})\mathbf A^T_{|\mathbf z|})^T \big]^T$ and the block $(\mathbf A_1 diag(\mathbf w_1)\mathbf A^T_1)^T$ is a $V \times V$ matrix.

The problem in (\ref{eq10a}) is combinatorial, and can  be
solved exactly by exhaustive search by computing $\lambda_2(\mathbf L')$ for ${|\mathbf z| \choose R}$ Laplacian matrices. However, this is not practical for large graphs that have large $|\mathbf z|$ and $R$. Instead, we are interested in proposing efficient heuristics for solving the problem in (\ref{eq10a}).

\vspace{-0.12in}
\subsection{Convex Relaxation}\label{CR}
The proposed SDP solution for multiple RISs in this subsection is mainly related to the preliminary work \cite{saifglobecom} but the authors considered a simple case of one RIS without utilizing the criticality of the nodes. 

The optimization vector in (\ref{eq10a}) is the vector $\mathbf z$. The $l$-th element of $\mathbf z$, denoted by $z_l$, is either
$1$ or $0$, which corresponds to whether this UE-RIS-UAV association should be chosen or not, respectively. Since (\ref{eq10a}) is NP-hard problem with high complexity, we relax the constraint on the entries of $\mathbf z$ and allow them to take any value in the interval $[0, 1]$. Specifically, we relax the Boolean constraint $\mathbf z \in \{0,1\}^{|\mathbf z|}$ to be a linear constraint $\mathbf z \in [0,1]^{|\mathbf z|}$, then we can represent the problem (\ref{eq10a}) as
\begin{align}
&\max_{\mathbf z} ~~~~~~~~~~~~~\lambda_2(\mathbf L'(\mathbf z))
\label{eq11a} \\
& {\rm subject~to\ } ~~~~~~\mathbf 1^T \mathbf z=R,\nonumber \\
& ~~~~~~~~~~~~~~~~~~~~0 \leq \mathbf z \leq 1.\nonumber
\end{align}
In \cite{4786516}, it was shown that $\lambda_2(\mathbf L'(\mathbf z))$ in (\ref{eq11a}) is the point-wise
infimum of a family of linear functions of $\mathbf z$, which is a concave function in $\mathbf z$. In addition, the relaxed constraints are linear in $\mathbf z$. Therefore, the optimization problem in (\ref{eq11a}) is a convex optimization problem \cite{4786516}, and is equivalent to the following SDP optimization problem \cite{CON}
\begin{align}
&\max_{\mathbf z, q} q
\label{eq100a} \\
& {\rm subject~to\ } q(\mathbf I - \frac{1}{|\mathbf z|}\mathbf 1 \mathbf 1^T) \preceq \mathbf L'(\mathbf z), \mathbf 1^T \mathbf z=R, 0 \leq \mathbf z \leq 1, \nonumber
\end{align}
where $\mathbf I \in \mathbf R^{V \times V}$ is the identity matrix and $\mathbf F \preceq \mathbf L$ denotes that $\mathbf L- \mathbf F$ is a positive semi-definite matrix.

The solution to the SDP optimization problem in (\ref{eq100a}) is explained as follows. First, we calculate the corresponding phase shifts of the RISs from each feasible UE node to each feasible UAV node, such that we generate all the possible schedules $\mathbf z$. In particular, the corresponding phase shift at PRU of the $r$-th RIS to reflect the signal of the $u$-th UE to the $a$-th UAV is calculated  as follows \cite{9293155}
\begin{multline} \label{phase} 
\theta^r_{m}=   \pi \frac{f_c}{c}\big\{ d_b(m_b-1)\psi^\text{RA}_{r,a}   \phi^\text{RA}_{r,a}+d_c(m_c-1)\psi^\text{RA}_{r,a} \varphi^\text{RA}_{r,a}\\+d_b(m_b-1)\psi^\text{UR}_{u,r}\phi^\text{UR}_{u,r}  +d_c(m_c-1)\psi^\text{UR}_{u,r}\varphi^\text{UR}_{u,r}\big\}.
\end{multline}
Second, we use off-the-shelf CVX software solver \cite{SDP-M} to  solve the SDP optimization problem in (\ref{eq100a}) and obtain $\mathbf z$.  The entries of the output vector $\mathbf z$ resulting from the CVX solver are continuous, that are between $0$ and $1$, and accordingly,  we consider to round the maximum $R$ entries to $1$ while others are rounded to zero. 

Note that if $n$ is an articulation node (i.e., the removal of that node disconnects the network \cite{4657335}), then  $\lambda_2(\mathcal G_{-n})$ theoretically equals to zero. This might cause numerical problem when we calculate the weight in the network connectivity. To avoid this problem, we introduce a small threshold, $\epsilon$. If $\lambda_2(\mathcal G_{-n}) \leq \epsilon$, we set $\lambda_2(\mathcal G_{-n})= \epsilon$. The steps of calculating the edge weights are summarized in Algorithm 1.

 \begin{algorithm}[t!]
	\caption{Weight Calculation Based On Nodes Criticality}
	\label{Algorithm1}
	\begin{algorithmic}[1]
		\State \textbf{Input:} Construct $\mathcal G(\mathcal V, \mathcal E)$ and set $\epsilon = 10^{-5}$
        \For{For each edge $l$ that connects two nodes $u,a$}
        \State Form the Laplacian matrix $\mathbf L(\mathcal G_{-u})$  as \eref{lap} and find the corresponding criticality value $\mathcal C_u$: 
        \If{$\lambda_2 (\mathcal G_{-u}) > \epsilon$}
         $\mathcal C_u=1/\lambda_2 (\mathcal G_{-u})$
         \Else ~~$\mathcal C_u=1/ \epsilon$
        \EndIf
        \State Form the Laplacian matrix  $\mathbf L(\mathcal G_{-a})$ as \eref{lap} and find the corresponding criticality value $\mathcal C_a$: 
        \If{$\lambda_2 (\mathcal G_{-a}) > \epsilon$}
         $\mathcal C_a=1/\lambda_2 (\mathcal G_{-a})$
         \Else ~~$\mathcal C_a=1/ \epsilon$
        \EndIf
        \State Calculate the weight of edge $l$ as \eref{wei}.
         \EndFor         
	\end{algorithmic}
\end{algorithm}

\subsection{A Greedy Perturbation Heuristic}\label{PH}
The SDP optimization has high complexity when $|\mathbf z|$ and $R$ are large, which is the case of large networks. 
Instead, we propose an effective greedy heuristic for solving (\ref{eq10a}) based on the values of the Fiedler vector, which is denoted by $\mathbf v$.  On the other hand, unlike the exhaustive search that calculates $\lambda_2(\mathbf L')$ for each possible association of UE-RIS-UAV, the proposed perturbation heuristic adds the $R$ edges one at a time by calculating only the weighted values of the differences between the values of the Fiedler vector. In the following proposition, we prove the upper bound of $\lambda_2(\mathbf L')$, and the proposed perturbation heuristic will be described next.


\begin{proposition}\label{prep1}
$\lambda_2(\mathbf L')$ is upper bounded by $ \lambda_2(\mathbf L)+ w_{l}(v_u-v_a)^2$, where  $v_u$ and $v_a$ are the corresponding values of the $u$-th and $a$-th indices of the Fiedler vector $\mathbf v$ of $\lambda_2 (\mathbf L)$ and $w_l$ is the weight of edge $l$ as given in \eref{wei}.

\end{proposition}
\begin{proof}    
For simplicity, we use $V_{ua}=(v_u-v_a)^2$.  If $\mathbf v$ is an eigenvector with unit norm corresponding to $\lambda_2(\mathbf L)$, then $\mathbf v \mathbf v^T$ is a supergradient of $\lambda_2(\mathbf L)$ \cite{4177113}. This means for any symmetric matrix $Y$ with size $V \times V$, we have 
\begin{equation}\label{tra}
\lambda_2(\mathbf L+Y)\leq \lambda_2(\mathbf L)+\textbf {Tr}(Y \mathbf v \mathbf v^T).
\end{equation}
In one connected graph,  $\lambda_2$ is isolated, where  $\lambda_1 < \lambda_2 < \lambda_3$, then $\lambda_2(\mathbf L')$ is an analytic function of $\mathbf L'$, and therefore of $\mathbf z$. In this case the supergradient is the gradient \cite{4177113}, i.e.,
\begin{equation} \label{der1}
   \pdv{\lambda_2(\mathbf L')}{z_l}=\mathbf v^T\pdv{\mathbf L'(\mathbf z)}{z_l}\mathbf v,
\end{equation}
where $\mathbf v$ is the unique normalized eigenvector corresponding to $\lambda_2(\mathbf L')$. By taking the partial derivative of 
\begin{equation}
\mathbf L'(\mathbf z)=\mathbf L+\sum^{|\mathbf z|}_{l=1} z_l w_{l} \mathbf a_l \mathbf a^T_l,
\end{equation}
we have 
\begin{equation} \label{der}
   \pdv{\mathbf L'(\mathbf z)}{z_l}=w_{l}\mathbf a_l \mathbf a_l^T.
\end{equation}
By substituting \eref{der} in \eref{der1}, we have 
\begin{align}   
 &\pdv{\lambda_2(\mathbf L')}{z_l}= w_{l} \mathbf v^T \mathbf a_l \mathbf a_l^T \mathbf v=  w_{l}V_{ua}.
\end{align}
Therefore, the partial derivative of $\lambda_2(\mathbf L')$ with respect to $z_l$ is $w_{l}V_{ua}$, where $l$ is the added edge between UE node $u$ and UAV node $a$. When $\lambda_2$ is isolated, $w_{l}V_{ua}$ gives the first order approximation of the increase in $\lambda_2(\mathbf L')$, if edge $l$ is added to the graph. Therefore, our step (2) of the greedy heuristic corresponds to adding an edge, from among the remaining $R$ edge candidates, that gives the largest possible increase in $\lambda_2(\mathbf L')$, according to a first order approximation. Therefore, we can say that if  $\mathbf v \mathbf v^T$ is a supergradient of $\lambda_2(\mathbf L')$ and based on \eref{tra}, $\lambda_2(\mathbf L')$ can be written as follows 
\begin{align} \label{weaker}
\lambda_2(\mathbf L') & =\lambda_2(\mathbf L+w_{l} \mathbf a_l \mathbf a_l^T)  \leq \lambda_2(\mathbf L)+ \textbf {Tr}( w_{l} \mathbf a_l \mathbf a_l^T \mathbf v \mathbf v^T) \\ & \nonumber 
\leq \lambda_2(\mathbf L)+ w_{l}  \textbf {Tr}( \mathbf a_l \mathbf a_l^T \mathbf v \mathbf v^T) 
\leq \lambda_2(\mathbf L)+ w_{l} V_{ua}.  
\end{align}
\eref{weaker} completes the proof.
\end{proof}

\textbf{Greedy Heuristic:}
Given \pref{prep1},  in each step of the proposed heuristic, we choose an edge $l$ that connects UE $u$ and UAV $a$, which has the largest value of $w_{l}V_{ua}$ that provides the maximum possible increase in $\lambda_2(\mathbf L')$. Starting from $\mathcal G$ and $\mathbf L$, we add new edges one at a time as follows:
\begin{itemize}
\item Calculate $\mathbf v$, a unit eigenvector corresponding to $\lambda_2 (\mathbf L)$, where $\mathbf L$ is the current Laplacian matrix.

\item From the remaining candidate edges corresponding to the UE-RIS-UAV schedules, add an
edge $l$ connecting UE $u$ and UAV $a$ with the largest $w_{l}V_{ua}$.

\item Remove all the UE-RIS-UAV candidate links of the already selected UE, RIS, UAV.

\end{itemize}
We stop the greedy heuristic when there is no feasible link to add. Since the number of UEs/UAVs is larger than the number of RISs, this heuristic stops when there is no more available RISs that have not been selected. The steps of the  greedy algorithm are given in Algorithm 2.

\begin{algorithm}[t!]
	\caption{The Proposed Greedy Perturbation Heuristic}
	\label{Algorithm2}
	\begin{algorithmic}[1]
		\State \textbf{Input:} UEs, UAVs, RISs, and network topology.
  \State Initially set $\mathbf L' \leftarrow \mathbf L$
        \For{$i=1, 2, \ldots, R$}
        \State Calculate $\mathbf v$ of the associated $\mathbf L'$. 
        \State From the remaining candidate edges corresponding to the UE-RIS-UAV schedules, add an
edge $l$ connecting UE $u$ and UAV $a$ with largest $w_{l}V_{ua}$, where $w_{l}$ 
is calculated as in Algorithm 1.
        \State Calculate the corresponding phase shift at PRU of the selected $r$-th RIS that reflects the signal of the $u$-th UE to the $a$-th UAV as  in \eref{phase}.
        \State Based on the selected edge $l$, update $\mathbf L'$.
         \State Remove all the UE-RIS-UAV candidate links of the already selected UE, RIS, UAV.
         \EndFor         
	\State \textbf{Output:} $\lambda_2(\mathbf L')$.
	\end{algorithmic}
\end{algorithm}

\section{Perturbation Heuristic Analysis}\label{LU}
In this section, we derive the upper and lower bounds of $\lambda_2 (\mathbf L')$ based on the proposed perturbation heuristic solution. Then, the computational complexity of the proposed schemes, as compared to the exhaustive search, is analyzed in  \sref{cc}.

\subsection{Lower and Upper Bounds Analysis}
Given the proposed perturbation heuristic, \pref{prep2} derives the lower and upper bounds on the algebraic connectivity of a graph obtained by adding $R$ edges connecting UE nodes to UAV nodes to a single connected graph. 
\begin{proposition}
\label{prep2}
Let $\mathbf L$ be the Laplacian matrix of the original connected graph $\mathcal G$.  Suppose we add an edge $l$  that connects UE node $u$ to UAV node $a$ through RIS $r$ to $\mathcal G$. Then, we have the following lower and upper bounds, respectively, 
for $\lambda_2 (\mathbf L+w_{l} \mathbf a_l \mathbf a_l^T)$:
\begin{align} \label{final_lower_bound_o} \nonumber 
& \lambda_2 (\mathbf L+w_{l}  \mathbf a_l \mathbf a_l^T)  \geq \lambda_2 \\ &  + \frac{w_l V_{ua}+\delta+2w_l -\sqrt{5w_lV_{ua}-w_l \delta^2+4w_l^2+4 w_l\delta}}{2},
\end{align}
\begin{align}   \label{upper_o}
\lambda_2 (\mathbf L+w_{l}  \mathbf a_l \mathbf a_l^T) \leq \lambda_2 + \frac{w_{l} V_{ua}} {1+ w_{l}  (2- V_{ua})/ (\lambda_n - \lambda_2)},
\end{align}
where $\delta= \lambda_3 - \lambda_2$.

\end{proposition}
\begin{proof}
Let $\mathbf L = \mathbf Q \mathbf D \mathbf Q^T$ be the eigenvalue decomposition of $\mathbf L$, where $\mathbf D$ is the diagonal matrix whose diagonal elements are the corresponding eigenvalues, denoted by $\lambda_1, \lambda_2, \ldots, \lambda_V$, and $\mathbf Q$ is an orthogonal matrix whose columns are the real, orthonormal eigenvectors of $\mathbf L$. Suppose that all entries in $\mathbf D$ are distinct (same process applies if eigenvalues other than $\lambda_2$ are repeated \cite{4177113}). Note that $\mathbf a_l \mathbf a_l^T$ is a matrix of rank-one and therefore our analysis follows the same steps used in \cite{Golub} for eigenvalues perturbation of a matrix with rank-one update. In particular,  the standard form in \cite{Golub} is
\begin{align} \label{updated_matrix}
\underbrace{\mathbf L'}_\text{updated matrix}= \mathbf Q\underbrace{\mathbf D}_\text{diagonal matrix with $\lambda_i$ entries} \mathbf Q^T+ \rho \underbrace{\mathbf a_l \mathbf a_l^T}_\text{rank-one},
\end{align}
where $\rho >0$, which will be replaced by $w_{l}$. Recall that $w_{l}$ is a non-negative value. Thus, we can write (\ref{updated_matrix}) as
\begin{align} \label{updated_matrix1}
\mathbf Q \mathbf L' \mathbf Q^T= \mathbf D + w_{l} (\mathbf Q \mathbf a_l) (\mathbf Q \mathbf a_l)^T.
\end{align}
We denote the eigenvalues of $\mathbf L'$ by $\tilde{\lambda}_1, \tilde{\lambda}_2, \ldots, \tilde{\lambda}_V$. Since we have one graph component, we assume  $\tilde{\lambda}_i \leq \tilde{\lambda}_{i+1}$, and similarly, $\lambda_i \leq \lambda_{i+1}$. Note that the matrices $\mathbf L$ and $\mathbf L'$ both have eigenvalue $0$ with the corresponding eigenvector $\mathbf 1$, i.e.,  $\lambda_1$ and $\tilde{\lambda}_1$ are zero.  Thus, we are interested in the remaining $V-1$ eigenvalues of $\mathbf L'$, i.e., particularly $\tilde{\lambda}_2$. Therefore, the eigenvalues of $\mathbf L'$ are the same as those of $\mathbf D + w_{l} \mathbf u \mathbf u^T$, where $\mathbf u = \mathbf Q \mathbf a_l$. To find the eigenvalues of $\mathbf L'$, assume first that $\mathbf D - \tilde{\lambda} \mathbf I$ is non-singular, we compute the characteristic polynomial as follows:
\begin{align}\nonumber 
\text {det}(\mathbf L'- \tilde{\lambda} \mathbf I)& = \text {det}(\mathbf D+ w_{l}\mathbf u \mathbf u^T - \tilde{\lambda} \mathbf I) \\ & \nonumber 
=   \text {det}((\mathbf D-\tilde{\lambda} \mathbf I)(\mathbf I+w_{l}(\mathbf D- \tilde{\lambda} \mathbf I)^{-1} \mathbf u \mathbf u^T))     \\ & \nonumber 
= \text {det}(\mathbf D-\tilde{\lambda} \mathbf I) \text {det}(\mathbf I+w_{l}(\mathbf D- \tilde{\lambda} \mathbf I)^{-1} \mathbf u \mathbf u^T).
\end{align}
Since $\mathbf D - \tilde{\lambda} \mathbf I$ is non-singular, $\text {det}(\mathbf I+w_{l}(\mathbf D- \tilde{\lambda} \mathbf I)^{-1} \mathbf u \mathbf u^T)=0$ whenever $\tilde{\lambda}$ is an eigenvalue. Note that $\mathbf I+w_{l}(\mathbf D- \tilde{\lambda} \mathbf I)^{-1} \mathbf u \mathbf u^T)$ is the identity plus rank-one matrix. The determinant of such matrix is as follows\footnote{By definition, if $\mathbf x$ and $\mathbf y$ are vectors, $\text {det}(\mathbf I + \mathbf x \mathbf y^T)= 1+\mathbf y^T \mathbf x$ \cite{book}.}:
\begin{align}\nonumber 
\text {det}(\mathbf I+w_{l}(\mathbf D- \tilde{\lambda} \mathbf I)^{-1} \mathbf u \mathbf u^T) &= (1 + w_{l}\mathbf u^T (\mathbf D - \tilde{\lambda} \mathbf I)^{-1}\mathbf u)\\&
= \underbrace{\big( 1+ w_{l} \sum_{i=2}^{V} \frac{u_i^2}{\lambda_i-\tilde{\lambda}}\big)}_\text{secular equation},
\end{align}
where $u_i=\mathbf Q_i^T \mathbf a_l$. Thus, $u_1=\mathbf Q_1^T \mathbf a_l=0$ since the values of the eigenvector corresponding to $\lambda_1$ is the same and $u_2= \mathbf Q_2^T \mathbf a_l= v_u-v_a$.
Golub \cite{Golub} showed that in the above situation the eigenvalues of $\mathbf L'$ are the zeros of the  secular equation, which is given as follows
\begin{align}\nonumber \label{secular}
&f(\tilde{\lambda})=1+w_{l} \sum_{i=2}^{V} \frac{u^2_i}{\lambda_i-\tilde{\lambda}}\\ &   
1+w_{l} \sum_{i=2}^{V} \frac{u^2_i}{\lambda_i-\tilde{\lambda}}=0 \to
w_{l} \sum_{i=2}^{V} \frac{u^2_i}{\tilde{\lambda}-\lambda_i}=1.
\end{align}

 \begin{figure}[t!]
	\centering
		\centerline{\includegraphics[width=0.78\linewidth]{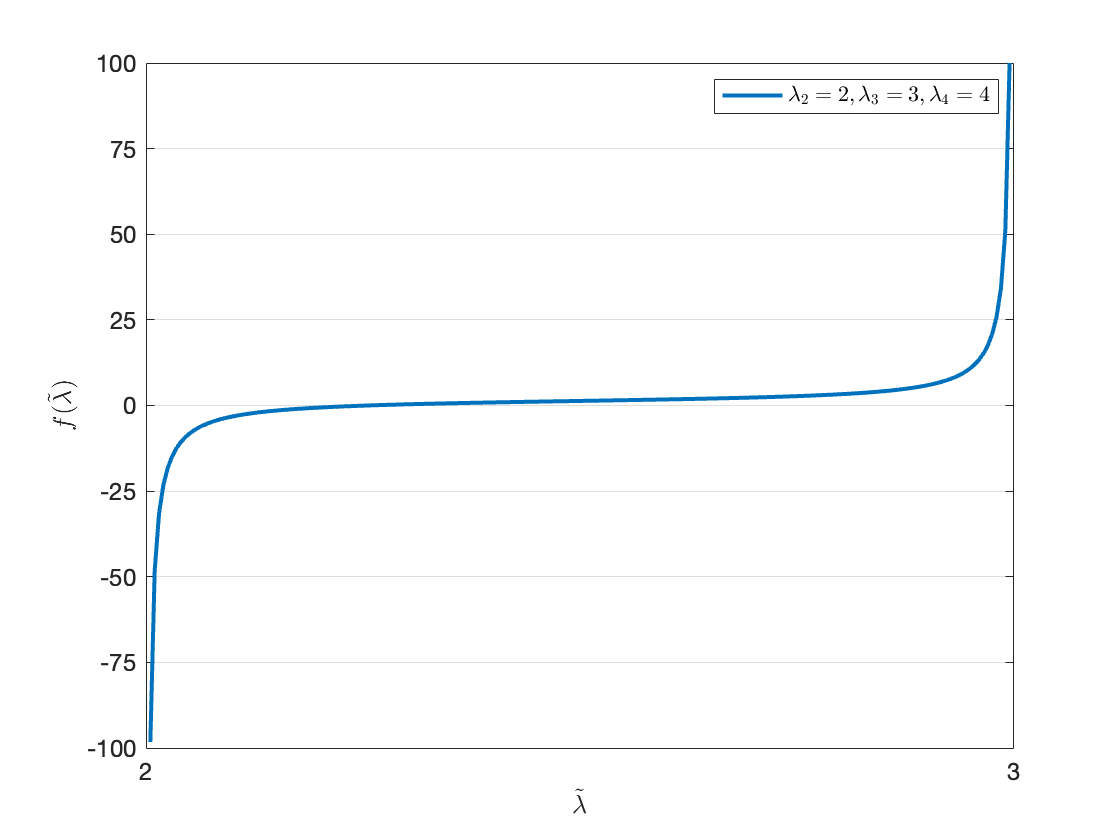}}
  \caption{ $f(\tilde{\lambda})= 1+\frac{0.5}{2-\tilde{\lambda}}+\frac{0.5}{3-\tilde{\lambda}}+\frac{0.5}{4-\tilde{\lambda}}$ in the interval $(\lambda_{2}, \lambda_{3})$ for $\lambda_i=i, u^2_i=0.5$, $i=2, 3, 4$ and $w_l=1$.}
	\label{fig2}
\end{figure}

Consider $u^2_i=0.5, \lambda_i=i, i=\{2,3,4\}$, $V=4$, and $w_l =1$, the function  $f(\tilde{\lambda})= 1+ \sum_{i=2}^{4} \frac{0.5}{\lambda_i-\tilde{\lambda}}$ has the graph shown in \Fref{fig2} for the interval $(\lambda_{2}, \lambda_{3})$. Since  $f'(\tilde{\lambda})=\sum_{i=2}^{4} \frac{u^2_i}{(\lambda_i-\tilde{\lambda})^2} > 0$, the function is strictly increasing in the interval  ($\lambda_i, \lambda_{i+1}$). Thus, the roots of  $f(\tilde{\lambda})$ are interlaced
by the $\lambda_i$. Since $\mathbf L'$ has more edges than $\mathbf L$ and by eigenvalue interlacing, we have \cite{Golub, JH}
\begin{align}\nonumber 
 & \lambda_i   \leq \tilde{\lambda}_{i} \leq \lambda_{i+1}, ~~~~~~~~\forall i =2, \ldots, V-1
\end{align}
 Let us consider $\tilde{\lambda}$ in the interval $(\lambda_{2}, \lambda_{3})$.  
For $\tilde{\lambda}$ in the interval $(\lambda_{2}, \lambda_{3})$, from Fig. 2,  if $f(\tilde{\lambda}) \leq 0$, then $\tilde{\lambda} \leq \tilde{\lambda}_2$. It is easy to see that $f(\tilde{\lambda}) \leq 0$ is equivalent to 
 \begin{align} \label{36}
w_{l} \frac{u^2_2}{\tilde{\lambda} -\lambda_2} \geq 1 +w_{l} \sum_{i=3}^{V} \frac{u^2_i}{\lambda_i-\tilde{\lambda}}.
 \end{align}
Thus, we conclude that (\ref{36}) implies $\tilde{\lambda} \leq \tilde{\lambda}_2$.

Recall
$\lVert  u \rVert = \lVert  \mathbf Q^T \mathbf a_l \rVert = \lVert   \mathbf a_l \rVert= 2$, $\sum_{i=3}^{V} u_i^2 \leq 2$ \cite{4177113}. From the right hand side of \eref{36}, we have
\begin{align}\label{37a}
& 1 +w_{l} \sum_{i=3}^{V} \frac{u^2_i}{\lambda_i-\tilde{\lambda}} \overset{(a)} \leq 1+   \frac{w_{l}}{\lambda_3 -\tilde{\lambda}}\sum_{i=3}^{V} u_i^2 \overset{(b)} \leq 1+ w_{l}  \frac{2}{\lambda_3 -\tilde{\lambda}},
 \end{align}
where $(a)$ and $(b)$ follow from the facts that  $\lambda_3 \leq \lambda_{4}, \ldots, \leq \lambda_{n}$ and $\sum_{i=3}^V u^2_i \leq 2$, respectively.
From \eref{37a}, if 
\begin{align}\label{37}
w_{l} \frac{u^2_2}{\tilde{\lambda}-\lambda_2} \geq 1+ w_{l}  \frac{2}{\lambda_3 -\tilde{\lambda}},
 \end{align}
 then $ \tilde{\lambda} \leq \tilde{\lambda}_2$. We set $\tilde{\lambda}-\lambda_2= \epsilon$ and $\lambda_3 -\lambda_2 = \delta$ in \eref{37}. We
aim to find $\epsilon >0$ (i.e., $\tilde{\lambda} > \lambda_2$) such that
\begin{equation} \label{epsilon}
w_{l} \frac{u^2_2}{\epsilon} \geq 1+ w_{l}  \frac{2}{\delta- \epsilon},
 \end{equation}
 By solving (\ref{epsilon}), we can verify that 
 \begin{align} 
 \epsilon= \frac{w_lu_2^2+\delta+2w_l-\sqrt{w_l(u_2^2-\delta^2)+4w_l^2+4w_l(u_2^2+\delta)}}{2},   
\end{align}
satisfies \eref{epsilon}. Thus,   the lower bound is
\begin{align} \label{final_lower_bound}
 \tilde{\lambda}_2& \geq \lambda_2  +\frac{w_lu_2^2+\delta+2w_l-\sqrt{5w_lu_2^2-w_l\delta^2+4w_l^2+4w_l\delta}}{2},   
\end{align}
where $u_2^2= V_{ua}$. This concludes that \eref{final_lower_bound} gives the lower bound in \eref{final_lower_bound_o}. 

Now, we derive the upper bound of $\tilde{\lambda}_2$. A sharper upper bound than \eref{weaker} can be obtained using the secular equation. From \eref{secular},   the algebraic connectivity of $\mathbf L'$ is the number $\tilde{\lambda}_2  \in (\lambda_2, \lambda_3)$ satisfying 
\begin{align}  \nonumber 
& w_{l} \frac{u^2_2}{\tilde{\lambda}-\lambda_2}+w_{l} \sum_{i=3}^{V} \frac{u^2_i}{\tilde{\lambda}-\lambda_i}=1  \\& \nonumber 
 w_{l} \frac{u^2_2}{\tilde{\lambda}-\lambda_2} = 1+ w_{l} \sum_{i=3}^{V} \frac{u^2_i}{\lambda_i - \tilde{\lambda}} \\&  \nonumber 
\frac{w_{l}}{\tilde{\lambda} - \lambda_2} = \frac{1+ w_{l} \sum_{i=3}^{V} \frac{u^2_i}{\lambda_i - \tilde{\lambda}}}{u^2_2} \\& \nonumber 
\tilde{\lambda} - \lambda_2 = \frac{w_{l} u^2_2} {1+ w_{l} \sum_{i=3}^{V} \frac{u^2_i}{\lambda_i - \tilde{\lambda}}} \\ & \nonumber 
\tilde{\lambda} = \lambda_2 + \frac{w_{l} u^2_2} {1+ w_{l} \sum_{i=3}^{V} \frac{u^2_i}{\lambda_i - \tilde{\lambda}}}
\end{align}

For the special case of $\tilde{\lambda}= \tilde{\lambda}_2$, we have
\begin{align}  \nonumber 
\tilde{\lambda}_2 = \lambda_2 + \frac{w_{l} u^2_2} {1+ w_{l} \sum_{i=3}^{V} \frac{u^2_i}{\lambda_i - \tilde{\lambda}_2}}.
\end{align}
Thus, the  upper bound of $\tilde{\lambda}_2$ is given as follows
\begin{align}  \nonumber \label{upper}
& \tilde{\lambda}_2 \overset{(a)}\leq  \lambda_2 + \frac{w_{l} u^2_2} {1+ w_{l} \sum_{i=3}^{V} \frac{u^2_i}{\lambda_i - \lambda_2}} \\ &  
\tilde{\lambda}_2 \overset{(b)}\leq  \lambda_2 + \frac{w_{l} u^2_2} {1+ w_{l}(\frac{1}{\lambda_n-\lambda_2}(\sum_{i=2}^{V} u^2_i - u^2_2))} \\& \nonumber 
\tilde{\lambda}_2 \overset{(c)}\leq \lambda_2 + \frac{w_{l} V_{ua}} {1+ w_{l}  (2- V_{ua})/ (\lambda_n - \lambda_2)},
\end{align}
where $(a)$, $(b)$, and $(c)$ come from $\lambda_2 < \tilde{\lambda}_2$, $\lambda_n \geq \lambda_{n-1}, \ldots, \geq \lambda_{2}$, $\sum_{i=2}^V u^2_i=2$, and $u_2^2= V_{ua}$, respectively. This concludes that \eref{upper} gives the upper bound in \eref{upper_o}.
\end{proof}

\subsection{Computational Complexity} \label{cc}
The exhaustive search requires a computational complexity of calculating $\lambda_2 (\mathbf L')$ for ${|\mathbf z| \choose R}$ Laplacian matrices, in which each  Laplacian matrix computation requires $\mathcal O\big(4 E' V^3/3\big)$ \cite{L}. Thus, exhaustive search requires $\mathcal O\big({|\mathbf z| \choose R} 4 E' V^3/3\big)$ operations. On the other hand, the SDP optimization  for the convex relaxation runs in high computational complexity for large $|\mathbf z|$. The proposed perturbation heuristic requires only an eigenvector value computation, as opposed to the exhaustive search. In particular, for the first added link, it computes the vector $\mathbf v$ of the current $\mathbf L$. Computing all the eigenvectors of an $V \times V$ dense matrix costs approximately $\mathcal O(4V^3/3)$ arithmetic operations \cite{4177113}. Since we have at maximum $R$ possible links to be added, the proposed perturbation solution runs in $\mathcal O(4RV^3/3+ 4R E' V^3/3)= \mathcal O(4RV^3/3(1+E')) \approx \mathcal O(4RV^3E'/3)$ arithmetic operations.

\section{Numerical Results}\label{NR}
We run MATLAB simulations to demonstrate the viability of the proposed schemes, and their superiority to the existing solutions. The simulation parameters of RISs configurations and UAV communications  are consistent with those used in \cite{8292633} and \cite{9293155}, respectively. We consider a RIS-assisted UAV system in an area of $150 ~m \times 150 ~m$, where the RISs have fixed locations and the UEs and the UAVs are distributed randomly. The RISs  are located at  an altitude of $20$ m,  $M=100$, $d_r=5$ cm, $d_c=5$ cm, $\beta_0=10^{-6}$, $N_0=-130$ dBm, the altitude of the UAVs is $50$ m, $f_c=3\times10^9$ Hz, $c=3\times 10^8$ m/s, $\alpha=4$, $p=1$ watt, $P=5$ watt, $\gamma_0^\text{(U)}=85$ dB, and $\gamma_0^\text{(A)}=80$ dB. Unless specified otherwise, $A=10$, $U=15$, $\gamma^\text{(RIS)}_0=30$ dB, and $\epsilon= 10^{-5}$.

The optimization problem in (\ref{eqob}) is solved using the two proposed heuristics in \sref{CR} and \sref{PH},  denoted by SDP and Proposed Perturbation, which are inspired by \cite{saifglobecom}, \cite{4177113}, respectively. For the sake of numerical comparison, the problem in (\ref{eqob}) is solved optimally via exhaustive search, which searches over all the feasible possible links between the UEs and the UAVs through the RISs, and then selects the maximum $\lambda_2(\mathbf L')$. In addition, we consider solving (\ref{eqob}) using the two benchmark schemes: original network without RISs deployment and random link selection.  Finally, we implement the upper and lower bounds, which are computed from \eref{upper_o} and \eref{final_lower_bound_o}, respectively.  Our performance measure is the network connectivity, which we calculate using $500$ iterations at each chosen value of UE,  UAV, RIS, and SNR threshold of the RISs. In each iteration, we change the locations of the UEs and the UAVs. 





In \Fref{fig3}, we plot the average network connectivity of the proposed and benchmark schemes versus the number of UEs $U$. From \Fref{fig3}, we can see that the proposed SDP and perturbation outperform the original scheme in terms of network connectivity. The results from the perturbation scheme are very close to the actual optimal value obtained using exhaustive search. This is because the perturbation scheme calculates the relative values of the Fiedler vector and selects the largest value that offers the possible maximum increase in the network connectivity, which is corresponding to  the desired UE-RIS-UAV link.  For this reason, the performance of the perturbation heuristic is also close to the upper bound. Due to controlling the phase shift of the RISs, the proposed schemes judiciously establish new links that connect the UEs to the desired UAVs such that the UEs do not miss the communications to the network while improving network connectivity.  Between these two proposed solutions, the  perturbation heuristic significantly outperforms SDP since the latter is sub-optimal. The original scheme that is without RISs deployment has poor performance. This interestingly shows  that by  adding a few number of low-cost passive nodes,   the average network connectivity of RIS-assisted UAV networks is improved significantly.  Notably, the values of $\lambda_2(\mathbf L')$  of all the schemes decreases as  the number of UEs increases, since adding more unconnected UEs may result in a sparse graph with low network connectivity. 

We observe from \Fref{fig3} that the performance of the proposed perturbation is very close to that of the optimal scheme. For ease of illustration,  we include the optimal scheme in \Fref{fig3} only and omit it in the remaining figures.

 \begin{figure}[t!]
	\centering
		\centerline{\includegraphics[width=0.95\linewidth]{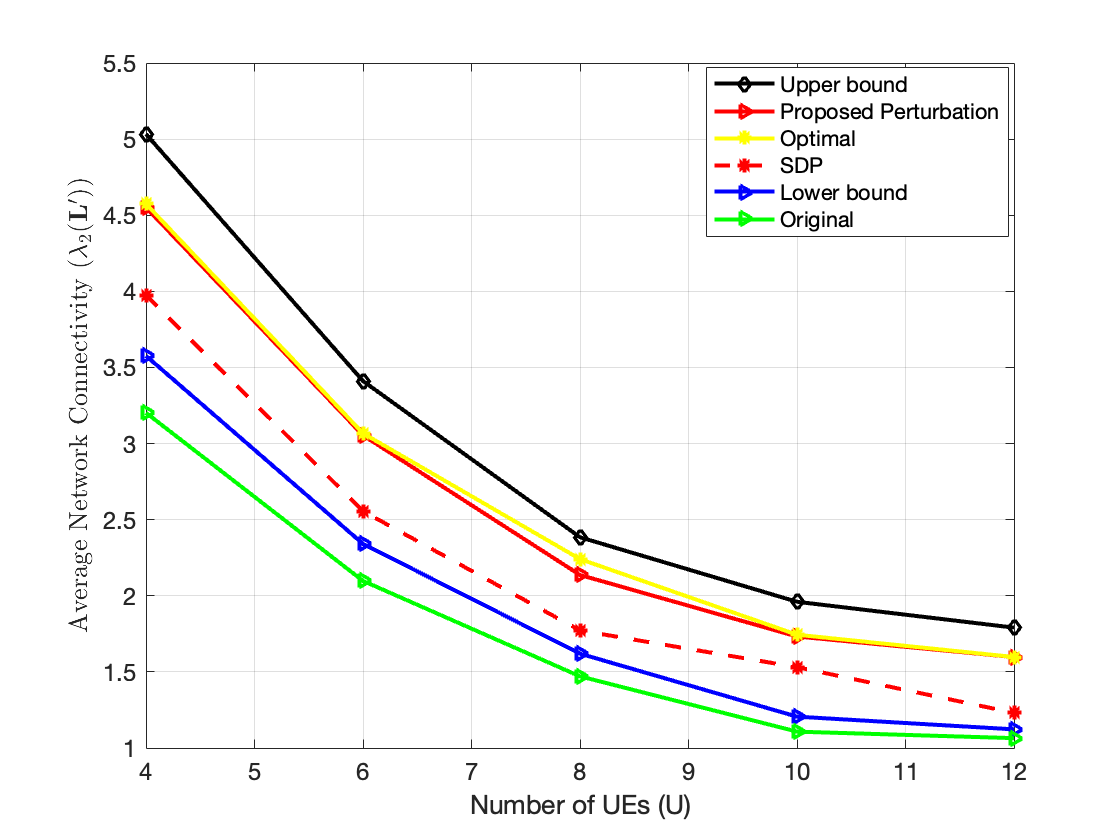}}
  \caption{The average network connectivity $\lambda_2(\mathbf L')$ versus the number of UEs $U$ for $7$ UAVs and $3$ RISs.}
	\label{fig3}
\end{figure}

In \Fref{fig4} and \Fref{fig5}, we show the average network connectivity of the proposed and benchmark schemes versus the  number of UAVs $A$ in different setups, i.e., small and large network sizes. For a small number of UAVs in \Fref{fig4}, the proposed SDP and perturbation schemes offer a slight performance gain in terms of  network connectivity compared to the original scheme. This is because our proposed schemes have a few options of UE-UAV links, where the RISs can direct the signal of the UEs to a few number of UAVs. However, when the number of UAVs increases, the proposed schemes smartly select  effective UE-RIS-UAV links  that significantly maximize the network connectivity. It is noted that $\lambda_2(\mathbf L')$ of all schemes increases  with the number of UAVs since adding more connected nodes to the network increases the number of new links, which increases the network connectivity. This also can be seen in \Fref{fig5} in the case of a large number of UAV nodes.

We observe that the values of $\lambda_2(\mathbf L')$ in \Fref{fig3} are smaller than the values of  $\lambda_2(\mathbf L')$ in \Fref{fig4} and \Fref{fig5} for all the UAV configurations. This is reasonable because adding more connected nodes of UAVs, adds more links to the network, thus improves the network connectivity  than adding more unconnected nodes of UEs.   The latter makes the network less connected (i.e., more UE nodes and no links between them).

 \begin{figure}[t!]
	\centering
		\centerline{\includegraphics[width=0.95\linewidth]{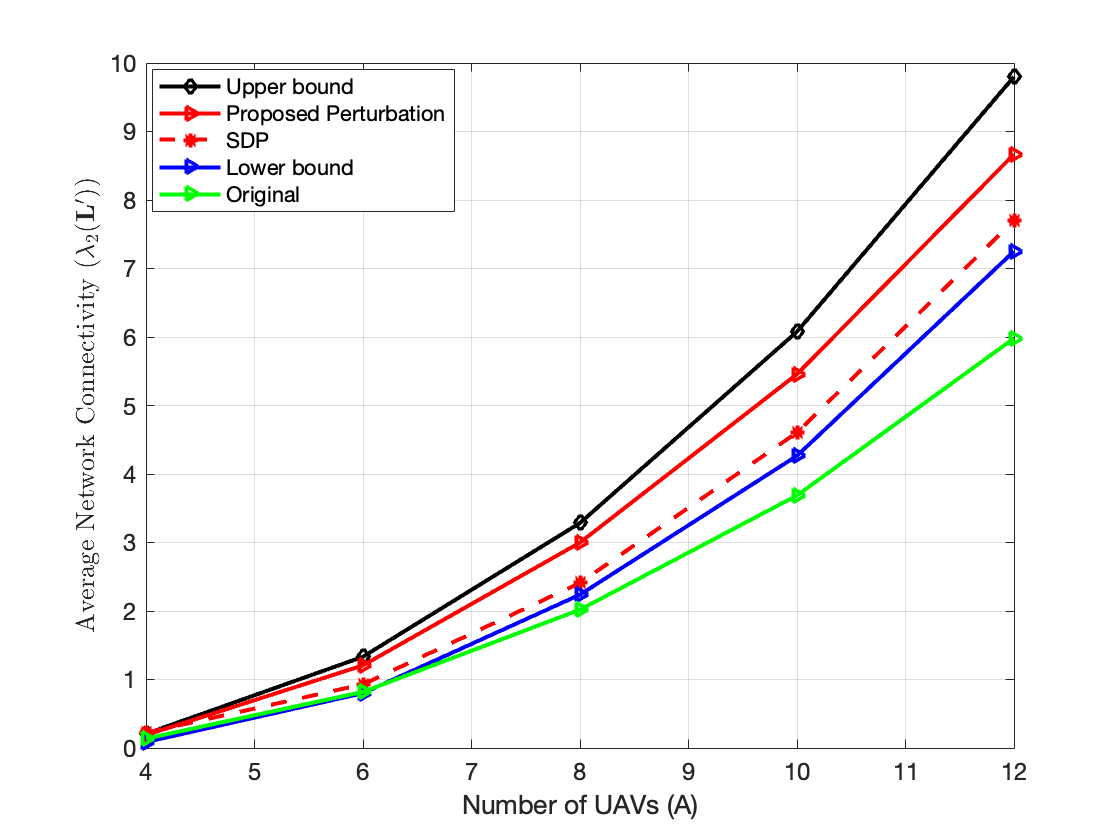}}
  \caption{The average network connectivity $\lambda_2(\mathbf L')$ versus the number of UAVs $A$ for $10$ UEs and $3$ RISs.}
	\label{fig4}
\end{figure}

 \begin{figure}[t!]
	\centering
		\centerline{\includegraphics[width=0.95\linewidth]{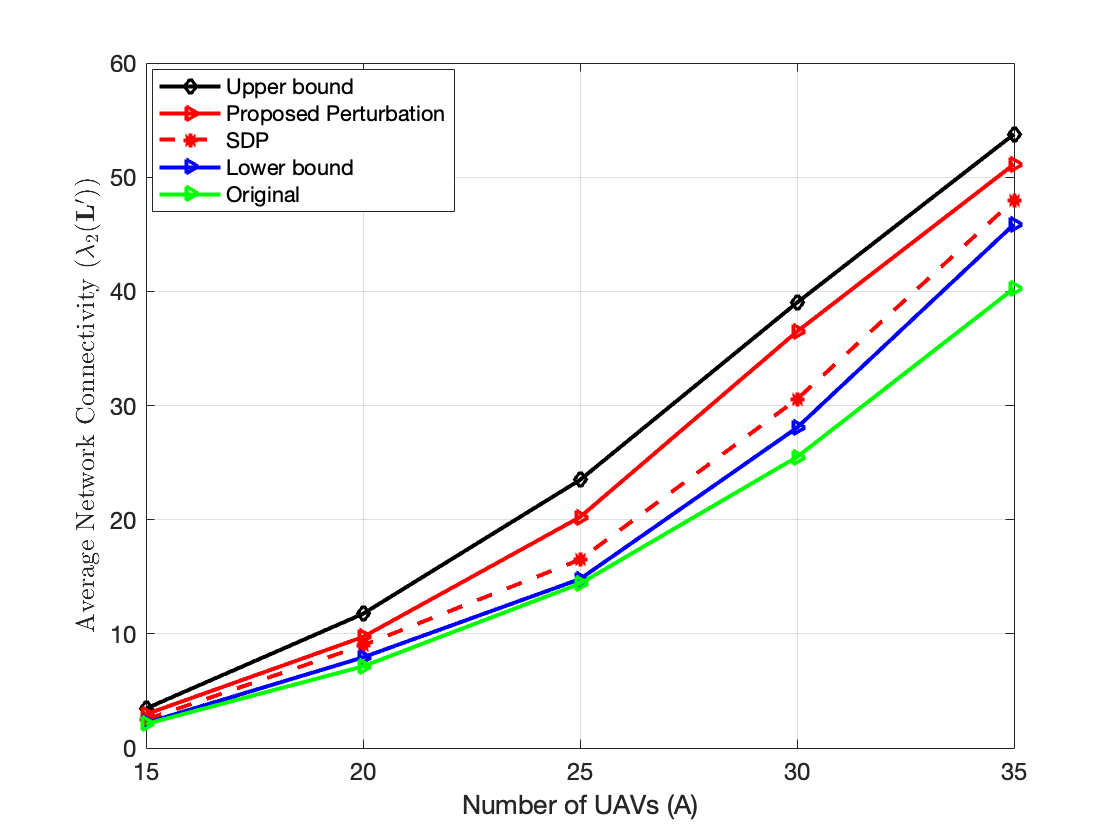}}
  \caption{The average network connectivity $\lambda_2(\mathbf L')$ versus a large number of UAVs $A$ for $22$ UEs and $3$ RISs.}
	\label{fig5}
\end{figure}

\begin{figure}[t!]
	\centering
		\centerline{\includegraphics[width=0.95\linewidth]{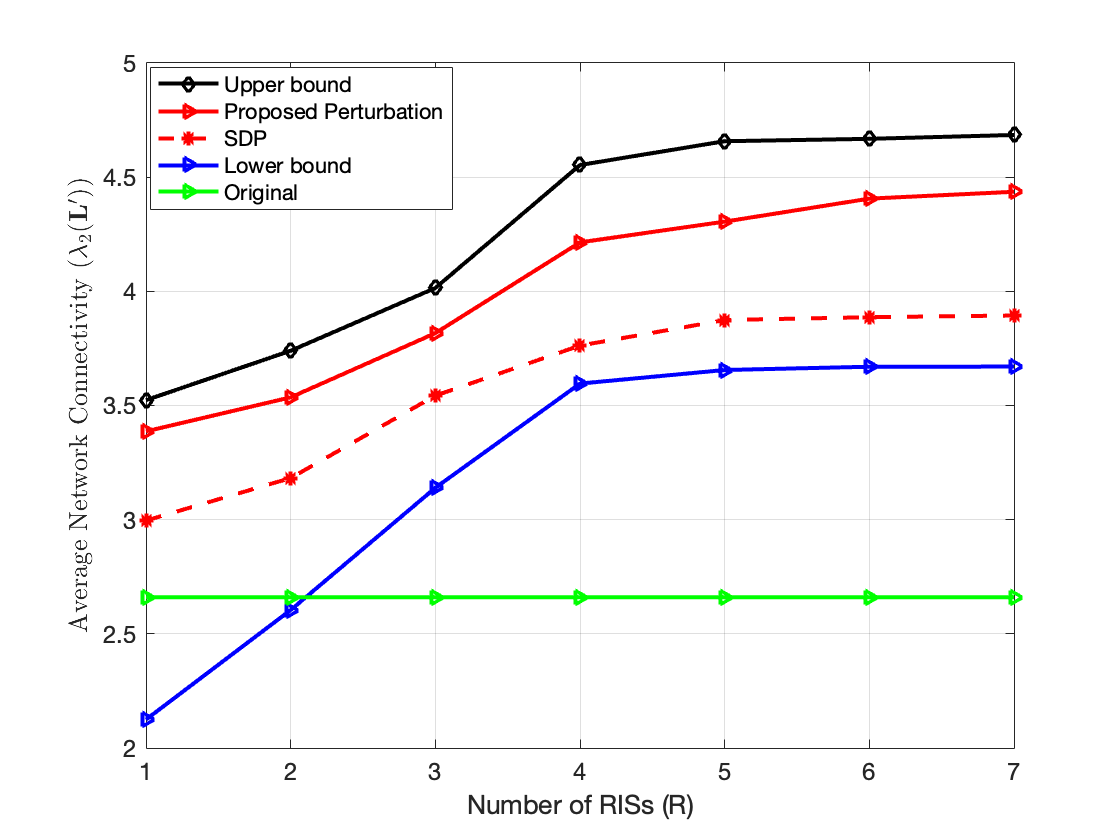}}
  \caption{The average network connectivity $\lambda_2(\mathbf L')$ versus the number of RISs $R$ for $15$ UEs and $10$ UAVs.}
	\label{fig6}
\end{figure}

To show that adding a few passive nodes is indeed crucial to maximize the network connectivity of UAV networks, \Fref{fig6}  plots the average network connectivity versus the number of RISs $R$.  For plotting this figure, we change the number of RISs and distribute them randomly and consider $15$ UEs and $10$ UAVs. For a small number of RISs in \Fref{fig6}, the performance of all schemes increases significantly since there are many possible selections of UEs and UAVs for each RIS, and therefore all the schemes select the best schedule that maximizes the network connectivity. However, when the number of RISs increases ($R>5$), there are no more good opportunities of selecting links that connect the remaining UEs  to the remaining UAVs. Thus, we notice a slight performance increase in the network connectivity of all the schemes. This also shows that adding the first few links is important to improve the network connectivity of the network. Note that the original scheme works irrespective of the RISs deployment, thus it has fixed performance  when we change the number of RISs.

To show the superior performance of  the proposed perturbation as compared to the random link addition, \Fref{fig7}  studies the average network connectivity versus the number of RISs $R$ for $15$ UEs and $10$ UAVs. For fair comparison, the random scheme is simulated  similar to the perturbation heuristic except that the selection of a link $l$ is randomly, not based on the maximum value of $w_lV_{ua}$. Thus, both schemes add the same number of links to the network. As expected, random addition performs poorly compared to the perturbation heuristic. However, the point to be highlighted here  is that a large increase in the average network connectivity can be obtained by adding a new edges carefully as we propose in the perturbation heuristic.

 \begin{figure}[t!]
	\centering
		\centerline{\includegraphics[width=0.95\linewidth]{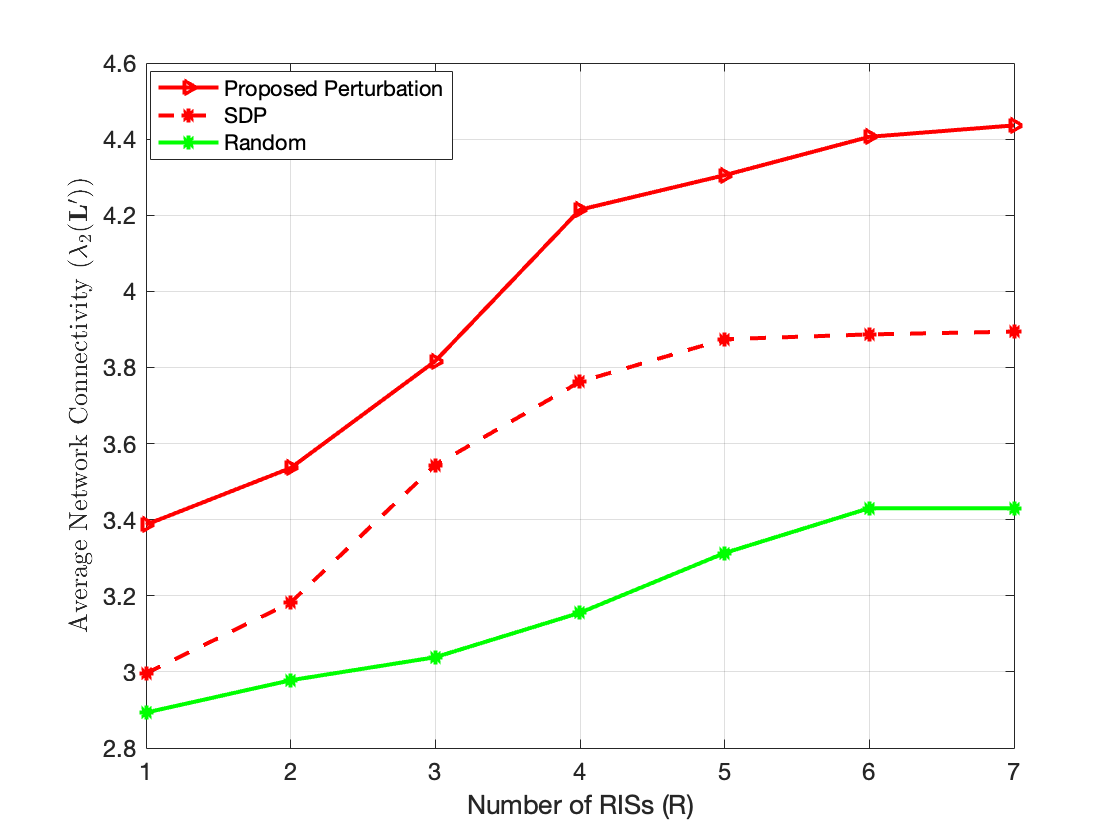}}
  \caption{The average network connectivity $\lambda_2(\mathbf L')$ versus the number of RISs $R$ for $15$ UEs and $10$ UAVs. A comparison of the perturbation heuristic with a random link addition.}
	\label{fig7}
\end{figure}

In \Fref{fig8}, we show the impact of the SNR threshold $\gamma_0^\text{(RIS)}$ on the network connectivity. For small SNR threshold,   all the links between the UEs and the UAVs through the RIS can satisfy this SNR threshold, thus many alternative links between the potential UE and the UAVs to select to maximize the network connectivity. On the other hand, for high RIS SNR threshold, a few UE-RIS-UAV links can satisfy such high SNR threshold. Thus, the network connectivity of all the schemes is degraded, and it becomes relatively close to the original scheme, which does not get affected by changing $\gamma_0^\text{(RIS)}$.

 \begin{figure}[t!]
	\centering
		\centerline{\includegraphics[width=0.95\linewidth]{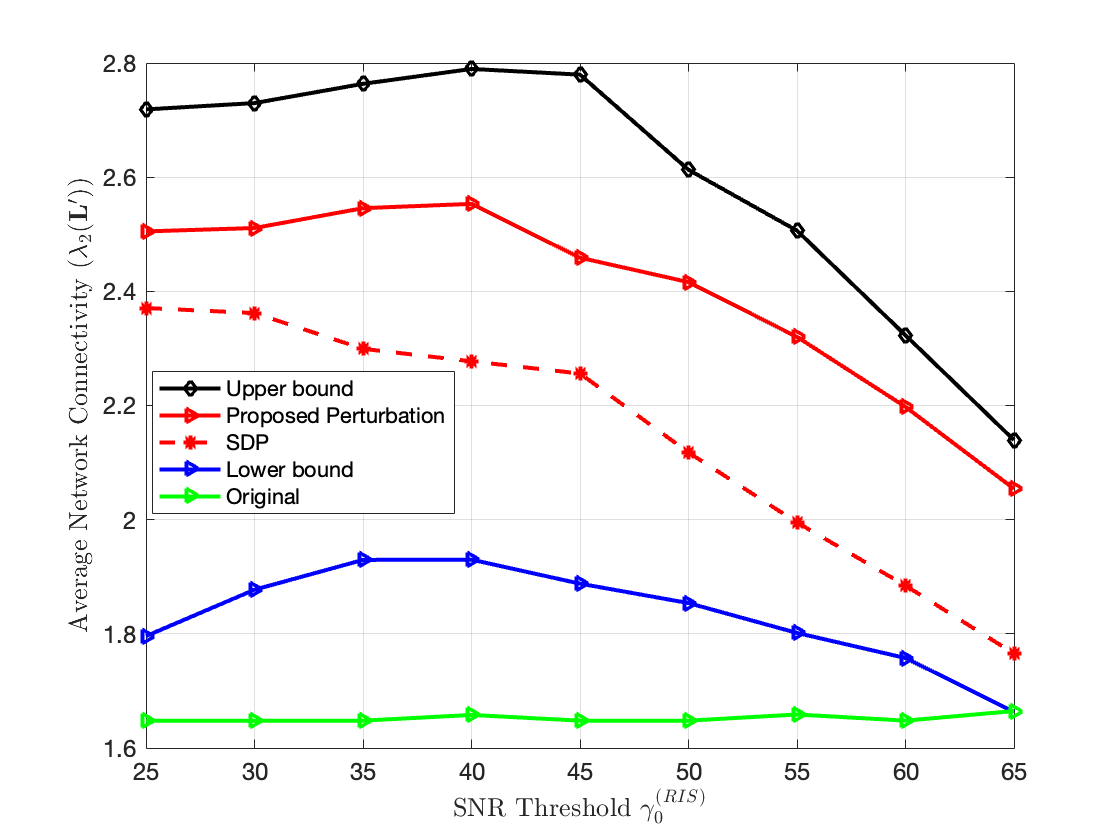}}
  \caption{The average network connectivity $\lambda_2(\mathbf L')$ versus SNR threshold $\gamma_0^\text{(RIS)}$ in dB for $12$ UEs, $8$ UAVs, and $3$ RISs.}
	\label{fig8}
\end{figure}

It is worth remarking that while the random scheme adds a random link to the network, the original scheme does not add a link. The proposed solutions balance between the aforementioned aspects by judiciously selecting an effective link, between a UE and a UAV,  that maximizes the network connectivity. This utilizes the benefits of the cooperation between an appropriate scheduling algorithm design and RISs phase shift configurations.  Compared to the optimal scheme, the proposed perturbation heuristic is near-optimal, however our proposed SDP has a certain degradation in network connectivity that comes as the
achieved polynomial computational complexity as compared to the high complexity of the optimal scheme.

\section{Conclusion}\label{C}
In this paper, we studied the effectiveness of RISs in UAV communications  in order to maximize the network connectivity by adjusting the algebraic connectivity using the reflected links of the RISs. We started by defining the nodes criticality and formulating the proposed network connectivity maximization problem. By embedding the nodes  criticality in the link selection, we proposed two  efficient solutions to solve the proposed problem. First, we proposed to relax the problem to a convex problem to be formulated  as an SDP optimization problem that can be solved efficiently in polynomial time using CVX. Then, we proposed a low-complexity, yet efficient, perturbation heuristic that adds one UE-RIS-UAV link at a time by calculating only the weighted difference between the values of the Fiedler eigenvector. We also derived the lower and upper bounds of the algebraic connectivity obtained by adding new links to the network based on the perturbation heuristic. Through numerical studies, we evaluated the performance of the proposed schemes via extensive simulations. We verified that the proposed schemes result in improved  network connectivity as compared to the existing solutions. In particular, the perturbation heuristic is near-optimal solution that shows very close results in terms of network connectivity compared to the optimal solution.


\begin{IEEEbiography}[{\includegraphics[width=1in,height=1.25in,clip,keepaspectratio, ]{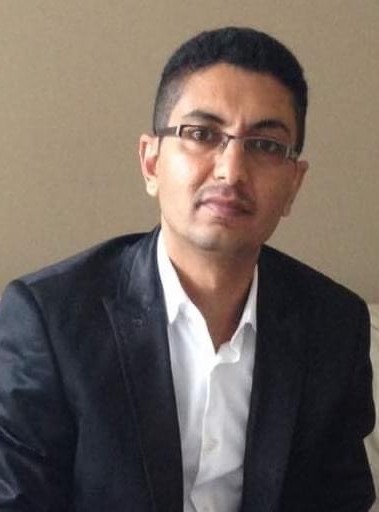}}]{Mohammed S. Al-Abiad} 
received the B.Eng. degree in computer and communication engineering from Taiz University, Taiz, Yemen, in 2010, the M.Sc. in electrical engineering from King Fahd University of Petroleum and Minerals (KFUPM), Dhahran, 
Saudi Arabia, in 2017,   the Ph.D. degree in electrical engineering from the University of British Columbia, BC, Canada, in 2020.  He is currently a Postdoctoral Research Fellow with the Edward S. Rogers Sr. Department of Electrical and Computer Engineering, University of Toronto. He was a Postdoctoral Research Fellow with the School of Engineering at the University of British Columbia, Canada, from 2020 to 2022. He was the recipient of the Natural Science and Engineering Research Council Postdoctoral Fellowship (NSERC PDF) of Canada in 2023. His research interests include  optimization and
resource allocation in wireless communications,  federated learning, mobile edge computing, and wireless networks connectivity using RISs. 
\end{IEEEbiography}

\begin{IEEEbiography}[{\includegraphics[width=1in,height=1.25in,clip,keepaspectratio, ]{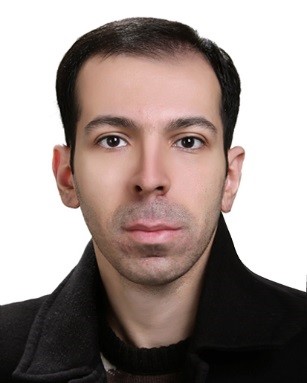}}]{Mohammed Javad-Kalbasi} 
received the M.Sc. degree in electrical and computer engineering from the Isfahan University of Technology in 2014. He is currently pursuing the Ph.D. degree with the University of Toronto. He is also a member of the Fujitsu Co-Creation Research Laboratory, University of Toronto. His main research interests include efficient Communications for the next generation of networks, information theory, scheduling in communications networks, and planning network migration. 
\end{IEEEbiography}

 \begin{IEEEbiography}[{\includegraphics[width=1in,height=1.25in,clip,keepaspectratio, ]{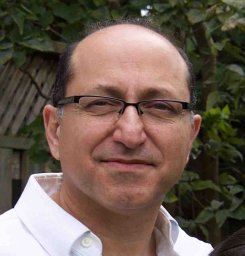}}]{Shahrokh Valaee} is a Professor with the Edward S. Rogers Sr. Department of Electrical and Computer Engineering, University of Toronto, and the holder of Nortel Chair of Network Architectures and Services. He is the Founder and the Director of the Wireless Innovation Research Laboratory (WIRLab) at the University of Toronto. Professor Valaee was the TPC Co-Chair and the Local Organization Chair of the IEEE Personal Mobile Indoor Radio Communication (PIMRC) Symposium 2011. He was the
TPC Co-Chair of ICT 2015, and PIMRC 2017, and the Track Co- Chair of WCNC 2014, PIMRC 2020, VTC Fall 2020. He is the co-chair of the organizing committee for PIMRC 2023. From December 2010 to December 2012, he was the Associate Editor of the IEEE Signal Processing Letters. From 2010 to 2015, he served as an Editor of IEEE Transactions on Wireless Communications. Currently, he is an Editor of the Journal of Computer and System Science and serves as a Distinguished Lecturer for IEEE Communication Society. He was the co-recipient of the best paper award in the IEEE Machine Learning for Signal Processing (MLSP) 2020 workshop. Professor Valaee is a Fellow of the Engineering Institute of Canada, and a Fellow of IEEE.
\end{IEEEbiography}

\end{document}